\documentclass[11pt]{article}
\usepackage{latexsym}
\usepackage{pict2e}\usepackage{epsfig}
\usepackage{amsmath,amsthm,amssymb}
\usepackage{color}
\usepackage[normalem]{ulem}
\usepackage{latexsym}
\usepackage{pict2e}
\usepackage{epsfig}
\usepackage{amsmath,amsthm,amssymb}
\usepackage{color}
\usepackage[normalem]{ulem}
\usepackage{subfigure}

\usepackage{algorithm}               
\usepackage{algorithmic}             
\usepackage{multirow}                
\usepackage{amsmath}
\usepackage{xcolor}

\usepackage{graphics}
\usepackage{graphicx}
\usepackage{epsfig}

\makeatletter 
\@addtoreset{equation}{section}
\makeatother  

\newtheorem{thm}{Theorem}
\newtheorem{lem}{Lemma}
\newtheorem{cor}{Corollary}
\newtheorem{Def}{Definition}

\usepackage{graphicx}
\graphicspath{{fig1/}}

\begin{document}

\begin{center}
\Large\bf The Component Diagnosability of General Networks\\
\vspace{0.5cm}
  \large
    Hongbin Zhuang$^{1,2}$, Wenzhong Guo$^{1}$, Xiaoyan Li$^{1,2*}$, Ximeng Liu$^{1,2}$, \\ Cheng-Kuan Lin$^{1}$ \\
    \small
    $^1$College of Computer and Data Science, Fuzhou University,\\ Fuzhou 350108, China\\
    $^2$Fujian Provincial Key Laboratory of Information Security of Network \\Systems, Fuzhou University, Fuzhou 350108, China\\
    $^*$Corresponding author:  xyli@fzu.edu.cn\\
\end{center}

\noindent\bf Abstract: \rm  The processor failures in a multiprocessor system have a negative impact on its distributed computing efficiency. Because of the rapid expansion of multiprocessor systems, the importance of fault diagnosis is becoming increasingly prominent. The $h$-component diagnosability of $G$, denoted by $c t_{h}(G)$, is the maximum number of nodes of the faulty set $F$ that is correctly identified in a system, and the number of components in $G-F$ is at least $h$. In this paper, we determine the $(h+1)$-component diagnosability of general networks under the PMC model and MM$^{*}$ model. As applications, the component diagnosability is explored for some well-known networks, including complete cubic networks, hierarchical cubic networks, generalized exchanged hypercubes, dual-cube-like networks, hierarchical hypercubes, Cayley graphs generated by transposition trees \textnormal{(}except star graphs\textnormal{)}, and DQcube as well. Furthermore, we provide some comparison results between the component diagnosability and other fault diagnosabilities.



\section{Introduction}
Multiprocessor systems have high-performance distributed computing capability, which makes them be able to carry out some coherent tasks. With the scale of the system continuously extends, the probability of processor failures increases dramatically so that we can't ignore such a phenomenon. How to sustain the performance and dependability of large-scale multiprocessor systems has been brought to the fore in the existence of processor failures~\cite{new2}. Thus, fault diagnosis is the most essential step in constructing and maintaining multiprocessor systems \cite{{guo,hong,liang17}}. Using a topological structure of the interconnection network, a multiprocessor system can be taken as an undirected graph which consists of many nodes acting as processors and many edges acting as communication links. In the following, we do not distinguish among multiprocessor systems, interconnection networks, and graphs.

The process of  identifying all faulty processors was referred to as \emph{system-level diagnosis}. There are several diagnosis models for system-level diagnosis. The \emph{PMC model}, presented by  Preparate, Metze, and Chien~\cite{pmc67}, is a test-based model. Under the PMC model, the processor sends test messages to adjacent processors for performing fault diagnosis. As a tester, a processor $x$ can diagnose its neighbor $y$ which is a testee. We call such a test as the ordered pair $\langle x, y\rangle$. Under the condition that the tester $x$ is fault-free, if the result of $\langle x, y\rangle$ is 0 (resp. 1), then the testee $y$ is fault-free (resp. faulty). When the tester $x$ is faulty, the result of $\langle x, y\rangle$ is unreliable. In \cite{mm81}, Maeng and Malek suggested another diagnosis model, the \emph{MM model}, which is a comparison-based model. In this model, a processor $z$ which we called as a comparator sends the same test to its two neighbours $x$ and $ y$, then compares their responses. We denote such a comparison as $(x,y)_{z}$.  $(x,y)_{z}=0$ (resp. 1) indicates that the test results for $x$ and $y$ are identical (resp. distinct). Assume that the comparator $z$ is fault-free. Then $(x,y)_{z}=0$ implies that both $x$ and $y$ are fault-free, while $(x,y)_{z}=1$ implies that there is at least one faulty processor between $x$ and $y$. When the comparator $z$ is faulty, the result of $(x,y)_{z}$ is unreliable. Based on the MM model, Sengupta and Dahbura \cite{sd92} proposed a special model, the \emph{MM$^{*}$ model}. This model supposes that every processor must test each pair of its adjacent processors.

A system $G$ is defined as \emph{$t$-diagnosable} if $|F| \leq t$ and all nodes in $F$ can be detected, where $F$ is the faulty node set of $G$. On the premise that all nodes of $F$ can be identified, the maximum number of nodes in $F$ is defined as the \emph{diagnosability} of the system. The diagnosability of many interconnection networks under the PMC model and MM$^{*}$ model have been studied \cite{{chiang12,cheng13,1lxy18}}. Moreover, many innovative concepts of fault diagnosability have emerged recently. Considering the actual situation, Lai et al.~\cite{ltch05} proposed a more realistic measure of diagnosability, \emph{conditional diagnosability}, which restricts that all adjacent nodes of every node can't fail simultaneously. Lin et al.~\cite{lin15} evaluated the conditional diagnosability of alternating group networks under the PMC model. As a generalization of conditional diagnosability, \emph{$h$-extra conditional diagnosability} was defined by Zhang et al. \cite{zhang16} which is the diagnosability under the condition that every component of $G-F$ has at least $h+1$ nodes where $F$ is the faulty node set of $G$. Li et al.~\cite{l19}~\cite{l20} studied the $h$-extra conditional diagnosability of the data center network DCell under the PMC model and MM$^{*}$ model. Zhang et al.~\cite{zhangs} determined the $h$-extra conditional diagnosability of twisted hypercubes under the MM$^{*}$ model. Peng et al.~\cite{plth12} proposed the \emph{$h$-good-neighbor conditional diagnosability}, which is the diagnosability under the assumption that every fault-free processor has at least $h$ fault-free neighbors. Zhao et al.~\cite{zhao19} established the $h$-good-neighbor conditional diagnosability of the hierarchical hypercube network under the PMC model and MM$^{*}$ model. Hu et al.~\cite{new5} investigated the equal relation between $h$-good-neighbor diagnosability under the PMC model and $h$-good-neighbor diagnosability under the MM$^{*}$ model of a graph.

If the faulty node set is large-scale, then there will be many components after deleting it from the multiprocessor system~\cite{{new1,new3,new4,new6}}. In this case, the diagnosability is closely related to the number of components. Given this, Zhang et al. \cite{zhang200} proposed a new fault diagnosability, called \emph{$h$-component diagnosability}. The $h$-component diagnosability $c t_{h}(G)$  of $G$ is the maximum number of nodes of the faulty set $F$ that are correctly identified in a system, and the number of components in $G-F$ is at least $h$. Our contributions are mainly as follows:

\begin{itemize}

\item  We determine the $(h+1)$-component diagnosability of the general network $G$ is $(h+1)(r-1)-\frac{h(h+1)}{2}+1$ with $r\geq 4$ and $1\leq h\leq r-3$ under the PMC model and MM$^{*}$ model, where the definition of $r$ is given in Theorem \ref{the1}.

\item  Using the results obtained, we establish the $(h+1)$-component diagnosability of some famous networks, including complete cubic networks, hierarchical cubic networks, generalized exchanged hypercubes, dual-cube-like networks, hierarchical hypercubes, Cayley graphs generated by transposition trees \textnormal{(}except star graphs\textnormal{)}, and DQcube. Moreover, we make some comparisons to show the advantages of component diagnosability.

\end{itemize}

The remaining sections of this paper are organized as follows. Section 2 lists the terms and notations used throughout the paper. Section 3 aims to determine the $(h+1)$-component diagnosability $c t_{h+1}(G)$ of general networks under the PMC model and MM$^{*}$ model. Section 4 applies the results in Section 3 to some famous networks. Section 5 provides some comparison results between the component diagnosability and other fault diagnosabilities. Finally, we conclude this paper in Section 6.

\section{Preliminaries}

\label{sec:Preliminaries}

In Section 2.1, we will introduce some terminologies and notations used in this paper. And then, in Section 2.2, we will present the concept of $h$-component diagnosability of a system $G$ under the PMC model and MM$^{*}$ model.

\subsection{Terminologies and notations}

For terminologies and notations not defined in this paper, we follow the reference~\cite{HsuBook}.  We represent a multiprocessor system by a simple undirected graph $G=(V(G),E(G))$, where $V(G)$ is a \emph{node set} consisting of all processors in $G$ and $E(G)=\{(u,v)|(u,v)$ is an unordered pair of $V(G)\}$ is an \emph{edge set}. If two nodes $u$ and $v$ are \emph{adjacent}, then $(u, v)\in E(G)$. The set $N_{G}(u)=\{v\in V(G)| (u,v)\in E(G)\}$ contains all the \emph{neighbors} of node $u$. Let $X\subseteq V(G)$, we use $G[X]$ to represent the subgraph of $G$ induced by the node subset $X$. We denote $G-X$ as $G[V(G)\backslash X]$. And we set $N_{G}(X)=\{v\in V(G)\setminus X|(u,v)\in E(G)$ and $u\in X\}=\bigcup\limits_{u\in X} N_{G}(u)\setminus X$ and $N_{G}[X]=N_{G}(X)\cup X$. Given two distinct node subsets $M, N\subseteq V(G)$, $E[M,N]$ is the set of all edges between $M$ and $N$ in $G$. A maximally connected subgraph of a graph is called \emph{component}. A component is odd if and only if the number of its nodes is odd, and is even otherwise. The number of odd components of $G$ is denoted by $o(G).$ The \emph{degree} of $u$ in $G$ is denoted by $deg_{G}(u)=|N_{G}(u)|$. Let $\delta(G)=\min\{deg_{G}(u)|u\in V(G)\}$, $\Delta(G)=\max\{deg_{G}(u)|u\in V(G)\}$. We use $G_{i}\cong G_{j}$ with $i\neq j$ to represent the graph $G_{i}$ is isomorphic to the graph $G_{j}$.

The $n$-dimensional hypercube $Q_{n}(n \geq 2)$ is a graph consisting of $2^{n}$ nodes, each of which has the form $u=$ $u_{n-1} u_{n-2} \cdots u_{0},$ where $u_{i} \in\{1,0\}$ for $0 \leq i \leq n-1.$
Two nodes $u=u_{n-1} u_{n-2} \cdots u_{0}$ and $v=v_{n-1} v_{n-2} \cdots v_{0}$ are adjacent if and only if there exists an integer $j \in$ $\{0,1, \cdots, n-1\}$ such that $u_{j} \neq v_{j}$ and $u_{i}=v_{i},$ for each $i \in\{0,1, \cdots, n-1\} \backslash\{j\} .$ Such an edge $(u, v)$ is called a $j$-dimensional edge.
There exist $V_{0}, V_{1}\subseteq V(Q_{n})$ such that the following two conditions hold:
\begin{itemize}
\item $V(Q_{n})=V_{0}\cup V_{1}$, $V_{0}\neq \emptyset$, $V_{1}\neq \emptyset$, $V_{0}\cap V_{1}=\emptyset$, and $Q_{n}[V_{0}], Q_{n}[V_{1}] \subseteq Q_{n}$;
\item $E(V_{0}, V_{1})$ is a perfect matching $M$ between $V_{0}$ and $V_{1}$ in $Q_{n}$.
\end{itemize}

We use $Q_{n-1}^{0}$, $Q_{n-1}^{1}$ to denote the induced subgraph $Q_{n}[V_{0}]$, $Q_{n}[V_{1}]$, respectively. Clearly, they are both $(n-1)$-dimensional hypercubes, and $E(Q_{n-1}^{0})$, $E(Q_{n-1}^{1})$, $M$ is a decomposition of $E(Q_{n})$. We define the decomposition as: $Q_{n}=G(Q_{n-1}^{0},Q_{n-1}^{1};M)$.

\begin{Def}\label{def1}\textnormal{~\cite{lhec12}}~ A set of nodes $S\subseteq V(G)$ is an $h$-component node cut if $G-S$ is disconnected and it has at least $h$ components.
\end{Def}

\begin{lem}\label{lem1}\textnormal{\cite{ja07}}
A graph $G$ has a perfect matching if and only if $o(G-A) \leq|A|$ for all $A \subseteq V(G)$.
\end{lem}

\begin{lem}\label{lem30}\textnormal{\cite{qz08}}
For any integer $n\geq 2$, any two nodes in $V(Q_{n})$ have exactly two common neighbors if they have.
\end{lem}

\subsection{The component diagnosability}

The set of test (comparison) outcomes under the PMC model and MM$^{*}$ model was referred to as the \emph{syndrome}. We use the notation $\Omega$ to represent the syndrome of the multiprocessor system. Let $\Omega(\widehat{F})$ represent the set of test (comparison) outcomes produced by the faulty set $\widehat{F}$. Define two different faulty sets of $V(G)$, $\widehat{F_1}$ and $\widehat{F_2}$. If $\Omega(\widehat{F_1})\cap\Omega(\widehat{F_2})=\emptyset$, $\widehat{F_1}$ and $\widehat{F_2}$ are called to be \emph{distinguishable}, that is, $(\widehat{F_1},\widehat{F_2})$ is a \emph{distinguishable pair}; otherwise, $\widehat{F_1}$ and $\widehat{F_2}$ are called to be \emph{indistinguishable} and $(\widehat{F_1},\widehat{F_2})$ is an \emph{indistinguishable pair}. The symmetric difference $(\widehat{F_1}\backslash \widehat{F_2})\cup (\widehat{F_2}\backslash \widehat{F_1})$ between $\widehat{F_1}$ and $\widehat{F_2}$ is denoted by $\widehat{F_1}\bigtriangleup \widehat{F_2}$. The sufficient and necessary condition for two faulty sets $\widehat{F_1}$ and $\widehat{F_2}$ are distinguishable under the PMC model and MM$^{*}$ model was proposed by Dahbura and Masson~\cite{Dahbura1984} and Sengupta and Dahbura~\cite{sd92}, respectively.

\begin{figure}[!htp]
  \centering
  \includegraphics[width=5in]{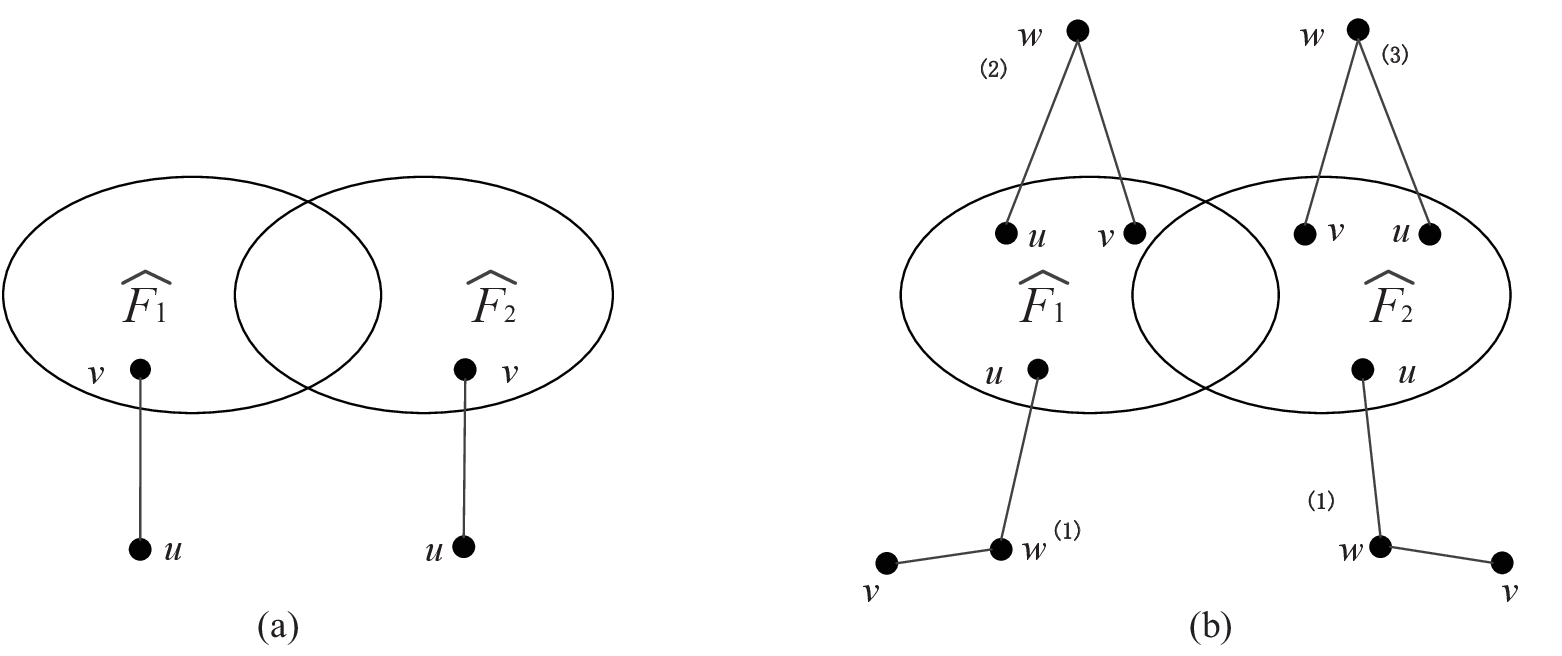}\\
  \caption{(a) An illustration for Lemma 3; (b) An illustration for Lemma 4.}\label{fig1}
\end{figure}

\begin{lem}\label{distinguishablepmc}\textnormal{\cite{Dahbura1984}}
Let $G=(V(G),E(G))$ be a multiprocessor system. For any two distinct sets $\widehat{F_1},\widehat{F_2}\subseteq V(G)$, $\widehat{F_1}$ and $\widehat{F_2}$ are distinguishable under the PMC model if and only if there exists at least one test from $V(G)\backslash (\widehat{F_1}\cup \widehat{F_2})$ to $\widehat{F_1}\bigtriangleup \widehat{F_2}$ \textnormal{(}see Fig.\!\ 1\textnormal{(}a\textnormal{))}.
\end{lem}

\begin{lem}\label{distinguishablemm}\textnormal{~\cite{sd92}} 
Let $G=(V(G), E(G))$ be a multiprocessor system. For any two distinct sets $\widehat{F_1}, \widehat{F_2}\subseteq V(G)$, $\widehat{F_1}$ and $\widehat{F_2}$ are distinguishable under the MM$^{*}$ model if and only if there is a node $w\in V(G)\backslash (\widehat{F_1}\cup \widehat{F_2})$ such that one of the following conditions holds \textnormal{(}see Fig.\!\ 1\textnormal{(}b\textnormal{)):}

\textnormal{(1)} $|N_{G}(w)- (\widehat{F_1}\cup \widehat{F_2})|\geq 1$ and $|N_{G}(w)\cap(\widehat{F_1}\Delta \widehat{F_2})|\geq 1$\textnormal{;}

\textnormal{(2)} $|N_{G}(w)\cap (\widehat{F_1}-\widehat{F_2})|\geq 2$\textnormal{;}

\textnormal{(3)} $|N_{G}(w)\cap (\widehat{F_2}-\widehat{F_1})|\geq 2$.
\end{lem}

The concept of $h$-component diagnosability of a system is presented as follows.

\begin{Def}~\label{def2}\textnormal{~\cite{zhang200}} \textnormal{(1)} Let $\widehat{F}\subseteq V(G)$ and $\widehat{F}$ be a fault-set. If $V(G)-\widehat{F}$ has at least $h$ components, then $\widehat{F}$ is called an $h$-component fault-set.

\textnormal{(2)}\;A system $G$ is $h$-component $t$-diagnosable if each distinct pair of $h$-component cuts $\widehat{F_1}$ and $\widehat{F_2}$ of $V(G)$ with $|\widehat{F_1}|\leq t$ and $|\widehat{F_2}|\leq t$ are distinguishable.

\textnormal{(3)} The $h$-component diagnosability, denoted by $ct_{h}(G)$, is defined as the maximum value of $t$ such that $G$ is $h$-component $t$-diagnosable.
\end{Def}

\section{Main results}
In this section, we will determine the $(h+1)$-component diagnosability of general networks.

\begin{thm}~\label{the1}
Let $r\geq 4$ and $1\leq h\leq r-3$. Let $G=(V(G), E(G))$ be a graph with a perfect matching and $|V(G)|\geq 2^{2r-2}$. Assume that $G$ satisfies the following two conditions\textnormal{:}

\textnormal{(a):}\ there exist a node $v$ and a set $A=\{v_{1}, v_{2}, \ldots, v_{h}, v_{h+1}\} \subseteq N_{G}(v)$ with $deg_{G}(x)=r$ for any node $x\in A \cup\{v\}$, such that $|N_{G}(v_{i_{1}})\cap N_{G}(v_{i_{2}})|=2$ $(1\leq i_{1}< i_{2}\leq h+1)$, $|N_{G}(v_{i_{1}})\cap N_{G}(v_{i_{2}})\cap \cdots \cap N_{G}(v_{i_{k}})|=1$ $(1\leq i_{1}< i_{2}<\cdots < i_{k}\leq h+1$ and $k\geq 3)$, and $|N_{G}(v)\cap N_{G}(v_{i})|=0$\textnormal{;}

\textnormal{(b):}\ for any subset $S\subseteq V(G)$ with $|S|\leq hr-\frac{(h-1)(h+2)}{2}-1$, $G-S$ is either connected or it has a component containing at least $|V(G)|-|S|-(h-1)$ nodes.

Then $c t_{h+1}(G)=(h+1)(r-1)-\frac{h(h+1)}{2}+1$ under the PMC model and MM$^{*}$ model.
\end{thm}
\begin{proof}
By condition (a), for $r\geq 4$ and $1\leq h\leq r-3$, there exist a node $v$ and a set $A=\{v_{1}, v_{2}, \ldots, v_{h}, v_{h+1}\} \subseteq N_{G}(v)$ with $deg_{G}(x)=r$ for any node $x\in A \cup\{v\}$.
Let $\widehat{F_1}=N_{G}(A)$ and $\widehat{F_2}=\widehat{F_1} \cup\{v_{h+1}\}$ (see Fig. 2).

\begin{figure}[!htp]
  \centering
  \includegraphics[width=3.5in]{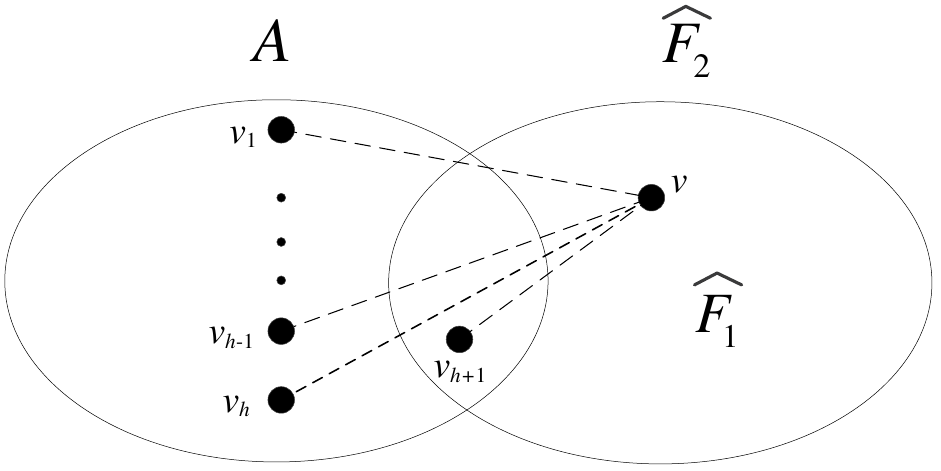}\\
  \caption{$\widehat{F_1}$ and $\widehat{F_2}$.}\label{fig2}
\end{figure}

By condition (a), we have
\[
\begin{aligned}
|\widehat{F_1}| &=(h+1)(r-1)-\binom{h+1}{2}+1 \\\
&=(h+1)(r-1)-\frac{h(h+1)}{2}+1.
\end{aligned}
\]

Then $|\widehat{F_2}|=(h+1)(r-1)-\frac{h(h+1)}{2}+2.$ Since $\widehat{F_1}$ consists of all neighbors of every node in $A$, $G-\widehat{F_1}$ has at least $|A|=h+1$ isolated nodes. Similarly, $G-\widehat{F_2}$ has at least $h$ isolated nodes. When $r\geq 4$ and $1\leq h\leq r-3$, $|V(G)|-|\widehat{F_2}|-h \geq2^{2r-2}-((h+1)(r-1)-\frac{h(h+1)}{2}+2)-h > 0.$ Then both $G-\widehat{F_1}$ and $G-\widehat{F_2}$ have at least $h+1$ components. Therefore, by Definition \ref{def1}, both $\widehat{F_1}$ and $\widehat{F_2}$ are $(h+1)$-component cuts. Since $\widehat{F_1} \triangle \widehat{F_2}=\{v_{h+1}\}$ and $N_{G}(v_{h+1}) \subseteq \widehat{F_1},$ there exists no edge between $\widehat{F_1} \triangle \widehat{F_2}$ and $\overline{\widehat{F_1} \cup \widehat{F_2}}.$ By Lemma \ref{distinguishablepmc} and Lemma \ref{distinguishablemm}, $\widehat{F_1}$ and $\widehat{F_2}$ are indistinguishable under the $\mathrm{PMC}$ model and $\mathrm{MM}^{*}$ model. By Definition \ref{def2}, $c t_{h+1}(G) \leq(h+1)(r-1)-\frac{h(h+1)}{2}+1.$ 

Now, we prove that $ct_{h+1}(G) \geq(h+1)(r-1)-\frac{h(h+1)}{2}+1$ holds. On the contrary, we assume that $ct_{h+1}(G) \leq(h+1)(r-1)-\frac{h(h+1)}{2}.$ That is, there exist two distinct $(h+1)$-component cuts $\widehat{F_1}$ and $\widehat{F_2}$ such that $|\widehat{F_1}|, |\widehat{F_2}| \leq(h+1)(r-1)-\frac{h(h+1)}{2}+1$ where $\widehat{F_1}$ and $\widehat{F_2}$ are indistinguishable. Without loss of generality, let
$\widehat{F_2} \backslash \widehat{F_1} \neq \emptyset$.

Note that $|\widehat{F_1}|, |\widehat{F_2}| \leq (h+1)(r-1)-\frac{h(h+1)}{2}+1 \leq (h+2)r-\frac{(h+1)(h+4)}{2}-1$ with $r\geq 4$ and $1\leq h\leq r-3$. Thus, by condition (b), $G-\widehat{F_2}$ has one large component $M$ plus a number of small components with at most $(h+2)-1=h + 1$ nodes. Let $W=V(G-\widehat{F_2}-M)$, then $|W| \leq h+1$. Then we have

\[
\begin{aligned}
|V(M)| &=|V(G)|-|\widehat{F_2}|-|W| \\
& \geq 2^{2r-2}-((h+1)(r-1)-\frac{h(h+1)}{2}+1)-(h+1).
\end{aligned}
\]

Let $f_{r}(h)=2^{2r-2}-((h+1)(r-1)-\frac{h(h+1)}{2}+1)-(h+1)$ with $r\geq 4$ and $1\leq h\leq r-3$. We can obtain that $\frac{\mathrm{d} f_{r}(h)}{\mathrm{d} h}=-r+h+ 1 / 2.$ Then $f_{r}(h)$ is an decreasing function when $h \leq r-3$. Thus $f_{r}(h) \geq f_{r}(r-3)=2^{2r-2}-r^{2}/2-r/2+2$. Then $|V(M)|\geq 2^{2r-2}-r^{2}/2-r/2+2$.

When $r \geq 4$ and $1\leq h\leq r-3$, the following inequality holds
\[
\begin{aligned}
|V(M)|-|\widehat{F_1}| & \geq2^{2r-2}-2((h+1)(r-1)-\frac{h(h+1)}{2}+1)-(h+1) > 0. \\
\end{aligned}
\]

\textbf{Claim1.} $V(M) \cap(\widehat{F_1} \backslash \widehat{F_2})=\emptyset$.

Since $|V(M)|-|\widehat{F_1}|>0$, we have $V(M) \backslash(\widehat{F_1} \backslash \widehat{F_2}) \neq \emptyset.$ On the contrary, we assume that $V(M) \cap(\widehat{F_1} \backslash \widehat{F_2}) \neq \emptyset.$ Let $L=V(M) \cap(\widehat{F_1} \backslash \widehat{F_2})$ and $P=V(M) \backslash(\widehat{F_1} \backslash \widehat{F_2}).$ Since $M$ is a connected component, there exist edges between $L$ and $P.$ Let $U=\{w \in P: w$ has a neighbor in $L\}$ and $O=P \backslash U$ (see Fig. 3).

Firstly, we prove that Claim 1 holds under the PMC model. As $L \subseteq \widehat{F_1} \backslash \widehat{F_2} \subseteq \widehat{F_1} \triangle \widehat{F_2} \text { and } P \subseteq \overline{\widehat{F_1} \cup \widehat{F_2}},$ there exist edges between $\widehat{F_1} \triangle \widehat{F_2}$ and $\overline{\widehat{F_1} \cup \widehat{F_2}}.$ By Definition \ref{def2}, $\widehat{F_1}$ and $\widehat{F_2}$ are distinguishable under the $\mathrm{PMC}$ model, a contradiction.

Next we prove that Claim 1 holds under the $\mathrm{MM}^{*}$ model. We assume that $O \neq \emptyset.$ Since $M$ is a connected component, there exists a node $w \in U$ such that $w$ has a neighbor $w^{\prime} \in O.$ It's easy to see that $w \in U$ also has a neighbor in $w^{\prime  \prime} \in L$ (see Fig.\;3(a)). By Definition 2, $\widehat{F_1}$ and $\widehat{F_2}$ are distinguishable under the $\mathrm{MM}^{*}$ model, a contradiction. Thus, we have $O=\emptyset$ and $U=P$.

\begin{figure}[!htp]
  \centering
  \includegraphics[width=3.5in]{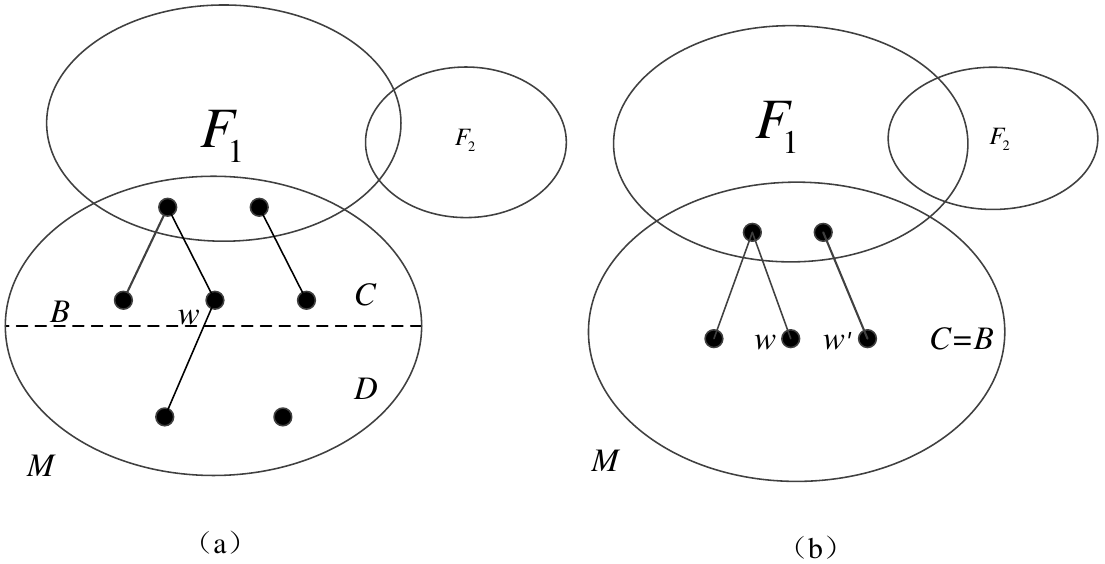}\\
  \caption{An illustration for the proof of Claim 1.}\label{fig3}
\end{figure}

Suppose that $w \in P$ has a neighbor $u \in P.$ By Definition \ref{def2}, $\widehat{F_1}$ and $\widehat{F_2}$ are distinguishable under the $\mathrm{MM}^{*}$ model, a contradiction. Then $P$ is an independent set of $G-(\widehat{F_1} \cup \widehat{F_2})$  (see Fig. 3(b)). By Lemma \ref{lem1}, $|P| \leq o(G-(\widehat{F_1} \cup \widehat{F_2})) \leq|\widehat{F_1} \cup \widehat{F_2}| \leq|\widehat{F_1}|+|\widehat{F_2}| .$ Hence, we have that
\[
|V(M)| \leq|\widehat{F_1} \backslash \widehat{F_2}|+|P| \leq 2|\widehat{F_1}|+|\widehat{F_2}| \leq 3(h+1)(r-1)-\frac{3h(h+1)}{2}+3.
\]
Then, we have
\[
2^{2r-2}-r^{2}/2-r/2+2 \leq|V(M)| \leq 3(h+1)(r-1)-\frac{3h(h+1)}{2}+3.
\]
This is a contradiction when $r \geq 4$ and $1\leq h\leq r-3.$ Then the claim holds. 

By Claim 1,  we can obtain that $\widehat{F_1} \backslash \widehat{F_2} \subseteq W$ and $|\widehat{F_1} \backslash \widehat{F_2}| \leq|W| \leq h+1 .$ Note that there exists no edge between $W$ and $V(M).$ Therefore, there exists no edge between $\widehat{F_1} \backslash \widehat{F_2}$ and $V(M).$ By the symmetry between $\widehat{F_1}$ and $\widehat{F_2},$ we can also obtain that $|\widehat{F_2} \backslash \widehat{F_1}| \leq h+1$ and there exists no edge between $\widehat{F_2} \backslash \widehat{F_1}$ and $V(M).$ Observe that $V(G-(\widehat{F_1} \cap \widehat{F_2}))$ consists of three parts $W, \widehat{F_2} \backslash \widehat{F_1}$ and $V(M) .$ Then $M$ is a component of $G-(\widehat{F_1} \cap \widehat{F_2})$ as well. When $r \geq 4$ and $1\leq h\leq r-3$, we have that

\[
\begin{aligned}
|V(M)|-|W|-|\widehat{F_2} \backslash \widehat{F_1}| & \geq(2^{2r-2}-r^{2}/2-r/2+2)-2(h+1) \\
& \geq(2^{2r-2}-r^{2}/2-r/2+2)-2(r-2) \\
&=2^{2r-2}-\frac{r^{2}}{2}-\frac{5r}{2}+6>0.
\end{aligned}
\]

Therefore, $M$ is the largest component of $G-(\widehat{F_1} \cap \widehat{F_2}).$ Recall that $W=V(G-\widehat{F_2}-M)$. Since $\widehat{F_2}$ is a $(h+1)$-component cut and $M$ is a connected component, the subgraph $G[W]$ contains at least $h+1-1=h$ components. Then $|W| \geq h$ and so $|W \cup(\widehat{F_2} \backslash \widehat{F_1})| \geq h+1.$ Then $G-(\widehat{F_1} \cap \widehat{F_2})$ has a large component $M$ and small components containing at least $h+1$ nodes.

Since $|\widehat{F_2} \backslash \widehat{F_1}| \geq 1$, $|\widehat{F_1} \cap \widehat{F_2}| \leq|\widehat{F_2}|-1 \leq (h+1)(r-1)-\frac{h(h+1)}{2}$. Let $g=h+1$, then $|\widehat{F_1} \cap \widehat{F_2}| \leq gr-\frac{(g-1)(g+2)}{2}-1$. By condition (b), $G-(\widehat{F_1} \cap \widehat{F_2})$ has a large component $M$ and a number of small components with at most $g-1=h$ nodes, a contradiction.
\end{proof}

\section{Applications to some well-known networks}
In Section 3, we show that the $(h+1)$-component diagnosability of general networks is $(h+1)(r-1)-\frac{h(h+1)}{2}+1$ under the PMC model and MM$^{*}$ model, where the definition of $r$ is given in Theorem \ref{the1}, i.e., $r$ is defined as the degree for each node of the node set satisfying the condition (a) in Theorem 1, the network is $r$-regular for the special case. In this section, we will apply the Theorem 1 to evaluate the $(h+1)$-component diagnosability of  some well-known networks, including complete cubic networks, hierarchical cubic networks, generalized exchanged hypercubes, dual-cube-like networks, hierarchical hypercubes, Cayley graphs generated by transposition trees \textnormal{(}except star graphs\textnormal{)}, and DQcube as well.

\subsection{The complete cubic network}
The complete cubic network is a famous network which generalizes the hierarchical cubic network. Therefore, we first review the definition and some available properties of the hierarchical cubic network as follows.

\begin{Def}\textnormal{~\cite{bakk12}}
An $n$-dimensional hierarchical cubic network $H\!C\!N_{n}$, consists of $2^{n}$ $n$-dimensional hypercubes $Q_{n}$, named clusters. Each node $u$ of $H\!C\!N_{n}$ is
uniquely associated with a pair of two $n$-bit sequences, where the first $n$-bit sequence identifies the cluster of $u$ and the second $n$-bit sequence identifies the position of $u$ inside its cluster. For one node $u=(c, p)\in H\!C\!N_{n}$, let $C_{c}$ be the cluster, which is an $n$-dimensional hypercube identified by $c$. Two nodes $(c, p)$ and $(d, q)$ are adjacent in $H\!C\!N_{n}$ if and only if one of the following conditions is satisfied \textnormal{(}where $\overline{x}$ is the binary complement of a bit sequence $x$\textnormal{)}.

\textnormal{(1)}\ if $c=d$, then $H(p,q)=1$\textnormal{;}

\textnormal{(2)}\ if $c\neq d$ and $c=p$, then $d=q=\overline{c}$\textnormal{;}

\textnormal{(3)}\ if $c\neq d$ and $c\neq p$, then $c=q$ and $p=d$\textnormal{,}
where $H(p,q)$ denotes the Hamming distance between two nodes $p$ and $q$.
\end{Def}

An $n$-dimensional hierarchical cubic network $H\!C\!N_{n}$ is $(n+1)$-regular with $2^{2n}$ nodes and $(n+1)2^{2n-1}$ edges. The edges of type (1) are referred to as cube edges, and the edges of types (2) and (3) are referred to as cross edges. For one pair of clusters $C_{c}$ and $C_{d}$, if $c=\overline{d}$, then there exist two cross edges between them; otherwise, there exists one cross edge. Clearly, there exists a perfect matching between clusters in $H\!C\!N_{n}$ (see $H\!C\!N_{2}$ in Fig. 4(a)).

The complete cubic network was introduced by Cheng et al.~\cite{cqs15}. In the following, we present the definition and some basic properties of complete cubic networks.

\begin{Def}\textnormal{~\cite{cqs15}}\label{def4}
The $n$-dimensional complete cubic network $CCN(n, f )$ for $n\geq 2$ is a collection of $2^{n}$ hypercubes $Q_{n}$, called clusters,
where the bijection function $f$ specifies a perfect matching on the nodes such that for any node $u$ in a cluster, $f(u)$ maps to another node in a different cluster. Besides the normal cube edges that connect all the nodes in each and every cluster, we refer to such an edge $(u, f(u))$ as a cross edge. Moreover, the matching $f$ needs to satisfy the following property that there exists a pairing $P$ of the clusters such that
\begin{itemize}
\item there are exactly two cross edges between two clusters $C_{c}$ and $C_{d}$ if they form a pair in $P$, and
\item there is exactly one cross edge between $C_{c}$ and $C_{d}$ if they do not form a pair in $P$.
\end{itemize}
\end{Def}

An $n$-dimensional complete cubic network $C\!C\!N_{n}$ is triangle-free and $(n+1)$-regular with $2^{2n}$ nodes and $(n+1)2^{2n-1}$ edges. Fig. 4(b)  shows the $2$-dimensional complete cubic network $C\!C\!N_{2}$.

\begin{figure}[htbp]
\centering
\subfigure[]{
\begin{minipage}[t]{0.5\linewidth}
\centering
\includegraphics[width=3in]{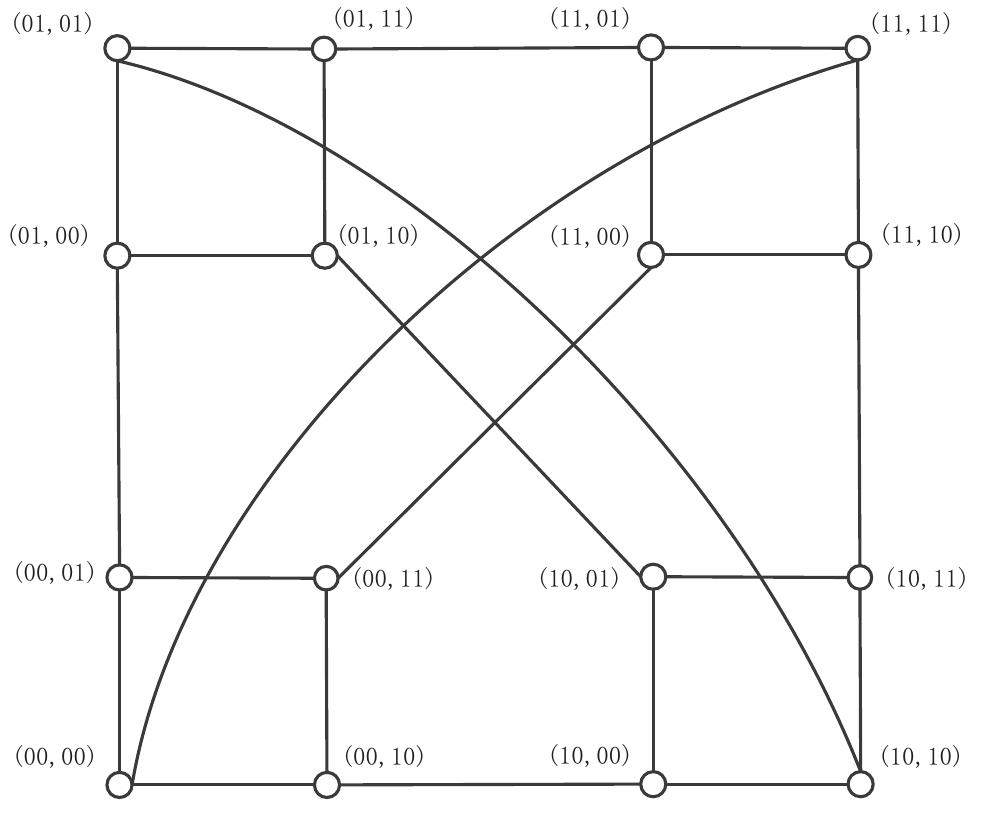}
\end{minipage}%
}%
\subfigure[]{
\begin{minipage}[t]{0.5\linewidth}
\centering
\includegraphics[width=3in]{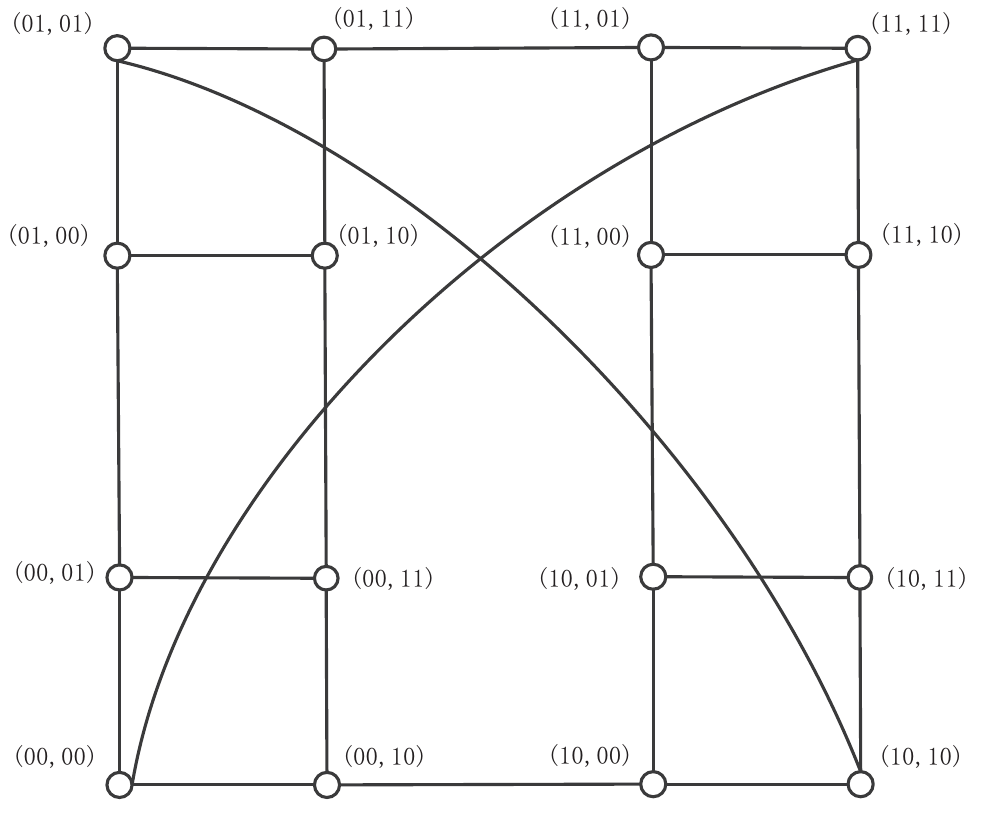}
\end{minipage}%
}%
\centering
\caption{(a)The illustration of $H\!C\!N_{2}$;  (b)The illustration of $C\!C\!N_{2}$.}
\end{figure}

\begin{lem}\textnormal{~\cite{cqs14}}\label{lem4}
For any integers $n\geq 2$ and $1\leq h\leq n$, for any $S\subseteq V(C\!C\!N_{n})$ with $|S|\leq (n+1)h-\frac{(h-1)(h+2)}{2}-1$, $C\!C\!N_{n}-S$ is either connected or it has a component containing at least $2^{2n}-|S|-(h-1)$ nodes.
\end{lem}

\begin{thm}
Let $n\geq 3$ and $1\leq h\leq n-2$, then the $(h+1)$-component diagnosability of $C\!C\!N_{n}$ is $c t_{h+1}(C\!C\!N_{n})=(h+1)n-\frac{h(h+1)}{2}+1$ under the PMC model and MM$^{*}$ model.
\end{thm}
\begin{proof}
By the definition of $C\!C\!N_{n}$, there exists a perfect matching between clusters in $C\!C\!N_{n}$ and $|V(C\!C\!N_{n})|=2^{2n}.$  $C\!C\!N_{n}$ consists of $2^{n}$ clusters, denoted as $C_{1}, C_{2}, \ldots, C_{2^{n}}$, and each cluster is isomorphic to $Q_{n}$. Since $C_{1}\cong Q_{n}$, let $C_{1}=G(X_{1}, Y_{1}; M)$ with $X_{1}\cong Q_{n-1}^{0}$ and $Y_{1}\cong Q_{n-1}^{1}$. Choosing an arbitrary node $v$ in $X_{1}$ and a node set $A=\{v_{1}, v_{2}, \ldots, v_{h}, v_{h+1}\} \subseteq N_{X_{1}}(v)$, we have that $deg_{C\!C\!N_{n}}(x)=n+1$ for any node $x\in A \cup\{v\}$. Note that $|V(C\!C\!N_{n})|$ is exactly equal to $2^{2n}$. Since $C_{1}\cong Q_{n}$, by Lemma \ref{lem30}, for any two nodes of $C_{1}$, if they have common neighbors, then they have exactly two common neighbors in $C_{1}$. We have $|N_{C_{1}}(v_{i_{1}})\cap N_{C_{1}}(v_{i_{2}})|=2$ $(1\leq i_{1}< i_{2}\leq h+1)$ and $|N_{C_{1}}(v_{i_{1}})\cap N_{C_{1}}(v_{i_{2}})\cap \cdots \cap N_{C_{1}}(v_{i_{k}})|=1$ $(1\leq i_{1}< i_{2}<\cdots < i_{k}\leq h+1$ and $k\geq 3)$. By the definition of $C\!C\!N_{n}$, there are two cross edges between one pair of clusters $C_{c}$ and $C_{d}$ if and only if $c=\overline{d}$; otherwise, there is only one cross edge, and cross edges have no common node. Then $|N_{C\!C\!N_{n}-C_{1}}(v_{i_{1}})\cap N_{C\!C\!N_{n}-C_{1}}(v_{i_{2}})|=0$ $(1\leq i_{1}< i_{2}\leq h+1)$. Thus, $|N_{C\!C\!N_{n}}(v_{i_{1}})\cap N_{C\!C\!N_{n}}(v_{i_{2}})|=2$ $(1\leq i_{1}< i_{2}\leq h+1)$ and $|N_{C\!C\!N_{n}}(v_{i_{1}})\cap N_{C\!C\!N_{n}}(v_{i_{2}})\cap \cdots \cap N_{C\!C\!N_{n}}(v_{i_{k}})|=1$ $(1\leq i_{1}< i_{2}<\cdots < i_{k}\leq h+1$ and $k\geq 3)$. Then $C\!C\!N_{n}$ satisfies the condition (a) of Theorem \ref{the1}. Moreover, by Lemma \ref{lem4}, $C\!C\!N_{n}$ satisfies the condition (b) of Theorem \ref{the1}. Thus, by Theorem \ref{the1}, $c t_{h+1}(C\!C\!N_{n})=(h+1)(r-1)-\frac{h(h+1)}{2}+1=(h+1)n-\frac{h(h+1)}{2}+1$ under the PMC model and MM$^{*}$ model.
\end{proof}

\begin{cor}
Let $n\geq 3$ and $1\leq h\leq n-2$, then the $(h+1)$-component diagnosability of $H\!C\!N_{n}$ is $c t_{h+1}(H\!C\!N_{n})=(h+1)n-\frac{h(h+1)}{2}+1$ under the PMC model and MM$^{*}$ model.
\end{cor}

\subsection{The generalized exchanged hypercube}
The exchanged hypercube is a link-diluted variation of the hypercube $Q_{s+t+1}$, proposed by Loh et al.~\cite{plwy05}. For a given positive integer $n$, let $I_{n}=\{1,2,\cdots,n\}$. The string $x_{n}x_{n-1}\cdots x_{1}$ is called a binary string of length $n$ if $x_{j}\in\{0,1\}$ for each $j\in I_{n}$. The definition of exchanged hypercubes are introduced as follows.

\begin{Def}\textnormal{~\cite{plwy05}}\label{def5}
Let $s,t\geq 1$, the exchanged hypercube $EH(s,t)$ consists of the node set $V(EH(s,t))$ and the edge set $E(EH(s,t))$, two nodes $u=u_{s+t}\cdots u_{t+1}u_{t}\cdots u_{1}u_{0}$ and $v=v_{s+t}\cdots v_{t+1}v_{t}\cdots$ $v_{1}v_{0}$ are linked by an edge, called $j$-dimensional edge, if and only if the following conditions are satisfied
\begin{itemize}
\item $u$ and $v$ differ exactly in one bit on the $j$-th bit or on the last bit,
\item if $j\in I_{t}$, then $u_{0}=v_{0}=1$,
\item $j\in I_{s+t}-I_{t}$, then $u_{0}=v_{0}=0$.
\end{itemize}
\end{Def}
Fig. 5(a) shows the regular exchanged hypercube $EH(1,1)$.

The generalized exchange hypercube, proposed by Cheng et al.~\cite{ekz17}. Let $s,t\geq 1$, the generalized exchanged hypercube $G\!E\!H(s,t,f)$ consists of the node set $V(G\!E\!H(s,t,f))$ and the edge set $E(G\!E\!H(s,t,f))$, where $|V(G\!E\!H(s,t,f))|=2^{s+t+1}$, $E(G\!E\!H(s,t,f))=E_{h}\cup E_{c}$, and $f$ is a one-to-one correspondence defined on $V(G\!E\!H(s,t,f))$.

$G\!E\!H(s,t,f)$ consists of two classes of hypercubes: one class contains $2^{t}$ $Q_{s}$'$s$, referred to as the Class-0 clusters; and the other contains $2^{s}$ $Q_{t}$'$s$, referred to as the Class-1 clusters. Class-0 and Class-1 clusters will be referred to as clusters of opposite class of each other, same class otherwise. The edges in the same cluster are referred to as cube edges, denoted by $E_{h}$, and the edges in different clusters are referred to as cross edges, denoted by $E_{c}$. There is exactly one cross edge between the clusters of opposite classes, and there is no edge between the clusters of same classes. Each node in Class-0 clusters has a unique neighbor in Class-1 clusters and vice versa. The function $f$ is a bijection between nodes of Class-0 clusters and those of Class-1 clusters, for two nodes $u, v$ in the same cluster, $f(u)$ and $f(v)$ are in two different clusters, the edge $(u, f(u))$ is a cross edge. The bijection $f$ ensures the existence of a perfect
matching between nodes of Class-0 clusters and those in the Class-1 clusters, but ignores the specifics of the perfect
matching. The generalized exchange hypercube $G\!E\!H(s,t)$ is triangle-free. Fig. 5(b) shows the irregular generalized exchanged hypercube $G\!E\!H(1,2)$.

In the special case when $s=t$, the exchanged hypercubes coincide with the so-called dual-cubes. The dual-cube $D_{n}$ was proposed by Li and Peng~\cite{ysp00}, which mitigates the problem of increasing number of links in the large-scale hypercube network. The dual-cube-like network $DC_{n}$~\cite{aacl13}, which is a generalization of dual-cubes, is isomorphic to $EH(n-1,n-1)$, a special case of $G\!E\!H(n-1,n-1)$ (see $DC_{3}$ in Fig. 6).

\begin{figure}[!ht]
 \begin{center}
 \resizebox*{5in}{!}{
 \includegraphics{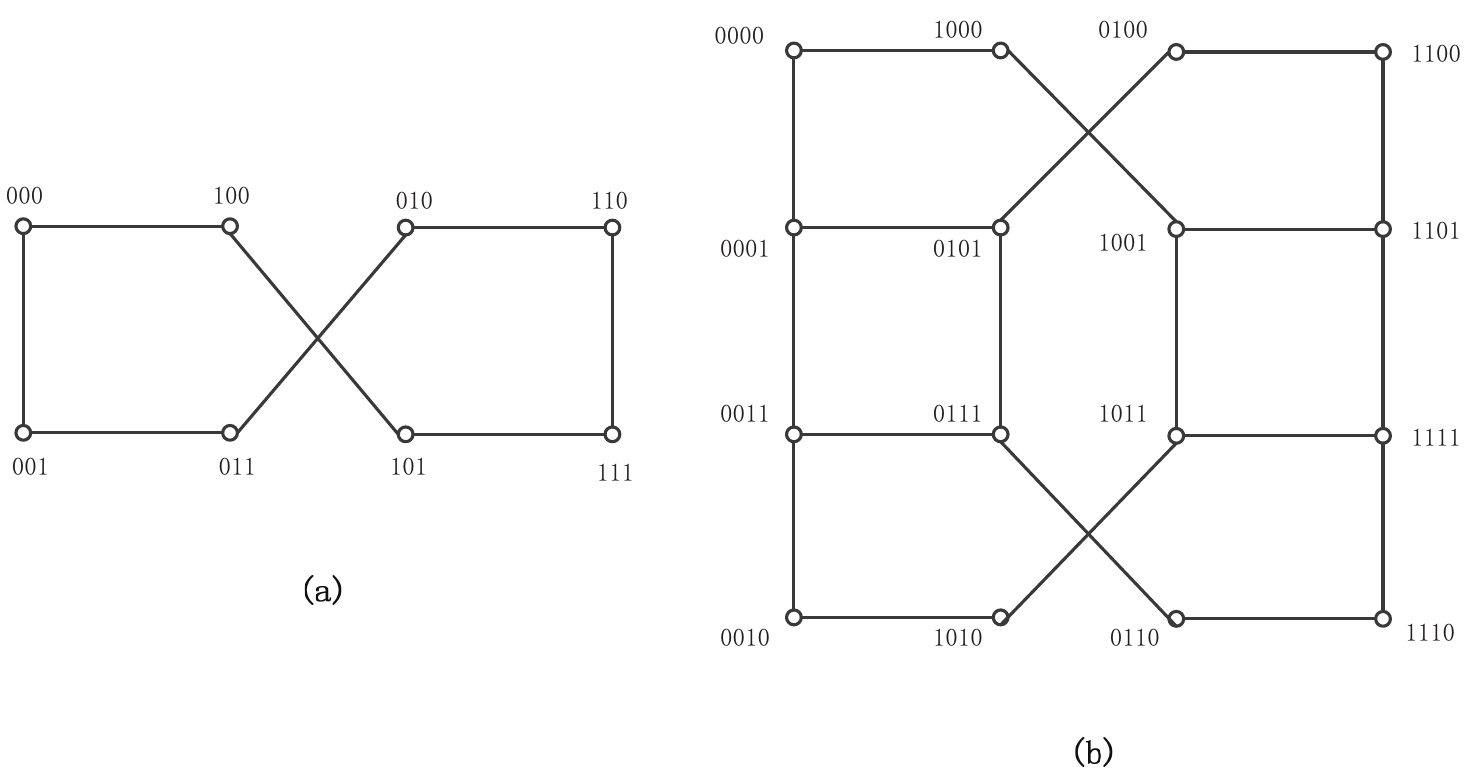}}
 \caption{(a)$EH(1,1)$; (b)$G\!E\!H(1,2)$.}
 \end{center}
\end{figure}

\begin{figure}[!ht]
 \begin{center}
 \resizebox*{5in}{!}{
 \includegraphics{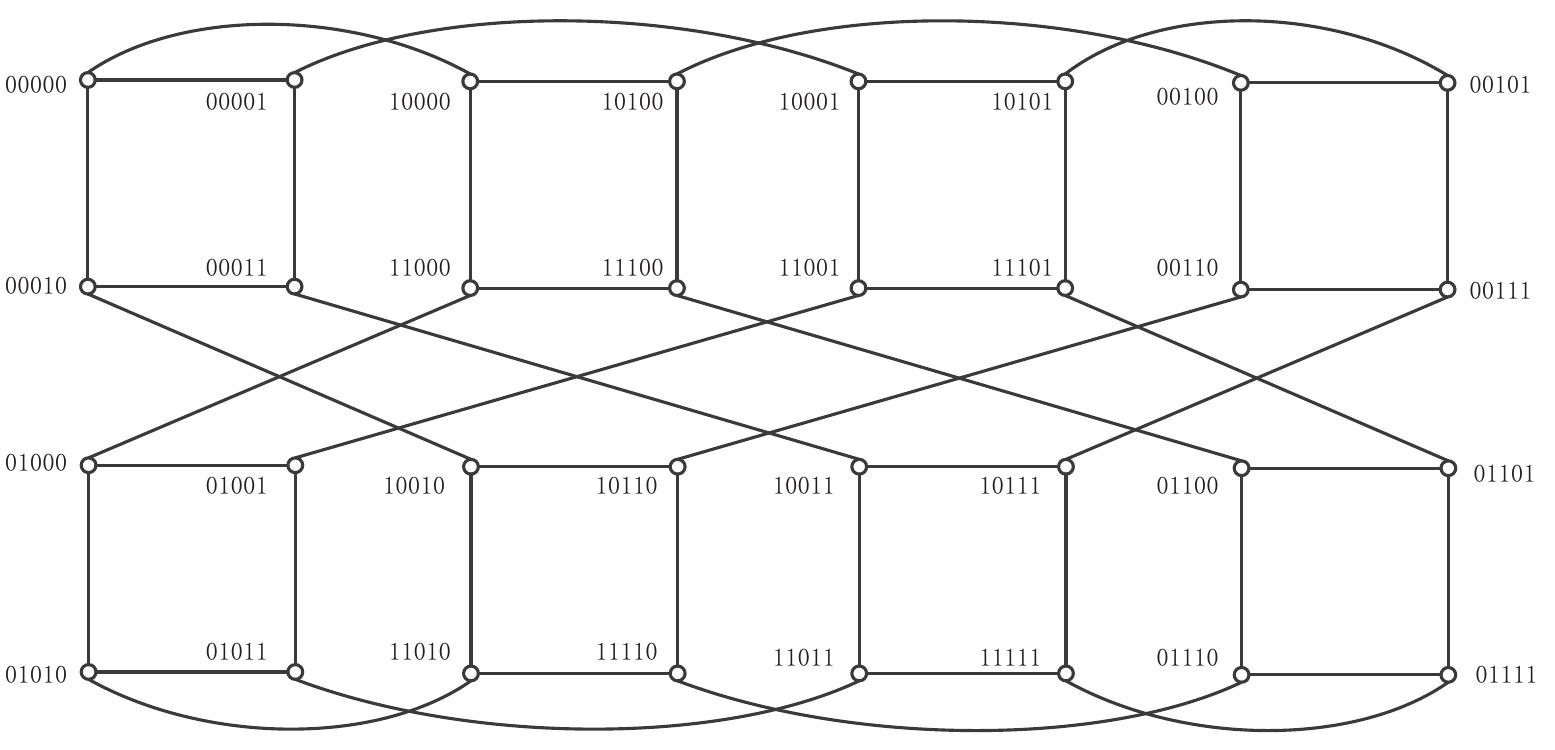}}
 \caption{The dual-cube-like network $DC_{3}$.}
 \end{center}
\end{figure}

\begin{lem}\textnormal{~\cite{eckz17}}\label{lem5}
For any integers $1\leq s\leq t$ and $1\leq h\leq s$, let $S\subseteq V(G\!E\!H(s,t))$ with $|S|\leq (s+1)h-\frac{(h-1)(h+2)}{2}-1$, then $G\!E\!H(s,t)-S$ is either connected or it has a component containing at least $2^{s+t+1}-|S|-(h-1)$ nodes.
\end{lem}

\begin{thm}
Let $3\leq s\leq t$ and $1\leq h\leq s-2$, then the $(h+1)$-component diagnosability of $G\!E\!H(s,t)$ is $c t_{h+1}(G\!E\!H(s,t))=(h+1)s-\frac{h(h+1)}{2}+1$ under the PMC model and MM$^{*}$ model.
\end{thm}
\begin{proof}
By Definition \ref{def5}, there exists a perfect matching between Class-0 clusters and Class-1 clusters in $G\!E\!H(s,t)$ and $|V(G\!E\!H(s,t))|=2^{s+t+1}$. $G\!E\!H(s,t)$ contains of $2^{s}+2^{t}$ clusters $C_{1},\ldots, C_{2^{t}},$
$C_{2^{t}+1},\ldots, C_{2^{s}+2^{t}}$ such that for $i\in [1,2^{t}]$, $C_{i}$, a Class-0 cluster, is isomorphic to $Q_{s}$, and for $i\in [2^{t}+1, 2^{s}+2^{t}]$, $C_{i}$, a Class-1 cluster, is isomorphic to $Q_{t}$.

Let $C_{1}=G(X_{1}, Y_{1}; M)$, a Class-0 cluster, that is isomorphic to $Q_{s}$ with $X_{1}\cong Q_{s-1}^{0}$ and $Y_{1}\cong Q_{s-1}^{1}$. Choosing an arbitrary node $v$ in $X_{1}$ and a node set $A=\{v_{1}, v_{2}, \ldots, v_{h}, v_{h+1}\} \subseteq N_{X_{1}}(v)$, we have that $deg_{G\!E\!H(s,t)}(x)=s+1$ for any node $x\in A \cup\{v\}$. For $3\leq s\leq t$, $|V(G\!E\!H(s,t))|=2^{s+t+1}> 2^{2s}$. Since $C_{1}\cong Q_{n}$, by Lemma \ref{lem30}, for any two nodes of $C_{1}$, if they have common neighbors at all, then they have exactly two common neighbors in $C_{1}$. We have $|N_{C_{1}}(v_{i_{1}})\cap N_{C_{1}}(v_{i_{2}})|=2$ $(1\leq i_{1}< i_{2}\leq h+1)$ and $|N_{C_{1}}(v_{i_{1}})\cap N_{C_{1}}(v_{i_{2}})\cap \cdots \cap N_{C_{1}}(v_{i_{k}})|=1$ $(1\leq i_{1}< i_{2}<\cdots < i_{k}\leq h+1$ and $k\geq 3)$. By the definition of $G\!E\!H(s,t)$, there is exactly one cross edge between the clusters of opposite classes, each node in Class-0 clusters has a unique neighbor in Class-1 clusters and vice versa, and cross edges have no common neighbor. Then $|N_{G\!E\!H(s,t)-C_{1}}(v_{i_{1}})\cap N_{G\!E\!H(s,t)-C_{1}}(v_{i_{2}})|=0$ $(1\leq i_{1}< i_{2}\leq h+1)$. Thus, $|N_{G\!E\!H(s,t)}(v_{i_{1}})\cap N_{G\!E\!H(s,t)}(v_{i_{2}})|=2$ $(1\leq i_{1}< i_{2}\leq h+1)$ and $|N_{G\!E\!H(s,t)}(v_{i_{1}})\cap N_{G\!E\!H(s,t)}(v_{i_{2}})\cap \cdots \cap N_{G\!E\!H(s,t)}(v_{i_{k}})|=1$ $(1\leq i_{1}< i_{2}<\cdots < i_{k}\leq h+1$ and $k\geq 3)$. Then $G\!E\!H(s,t)$ satisfies the condition (a) of Theorem \ref{the1}. Moreover, by Lemma \ref{lem5}, $G\!E\!H(s,t)$ satisfies the condition (b)  of Theorem \ref{the1}. Thus, by Theorem \ref{the1}, $c t_{h+1}(G\!E\!H(s,t))=(h+1)(r-1)-\frac{h(h+1)}{2}+1=(h+1)s-\frac{h(h+1)}{2}+1$ under the PMC model and MM$^{*}$ model.
\end{proof}

\begin{cor}
Let $n\geq 4$ and $1\leq h\leq n-3$, then the $(h+1)$-component diagnosability of $DC_{n}$ is $c t_{h+1}(DC_{n})=(h+1)(n-1)-\frac{h(h+1)}{2}+1$ under the PMC model and MM$^{*}$ model.
\end{cor}

\subsection{The hierarchical hypercube}
The hierarchical hypercube was proposed by Malluhi and Bayoumi~\cite{qmma94}, which is a modification of an $n$-dimensional cube-connected-cycle $C\!C\!C_{n}$~\cite{fpjv81}, and the cycle is replaced with a hypercube. The definition and some available properties of hierarchical hypercubes are introduced as follows.

\begin{Def}\textnormal{~\cite{qmma94}}\label{def6}
The $n$-dimensional hierarchical hypercube, $H\!H\!C_{n}$, is defined to be a graph with node set $\{(X,Y)|X=a_{n-1}a_{n-2}\ldots a_{m}, Y=a_{m-1}a_{m-2}\ldots a_{0}$, and $a_{i}\in\{0,1\}$ for all $0\leq i\leq n-1\}$, where $n=2^{m}+m$ and $m\geq 1$. Node adjacency of $H\!H\!C_{n}$ is defined as follows: $(A,B)$ is adjacent to

\textnormal{(1)}$(A,B^{l})$ for all $0\leq l\leq m-1$, and

\textnormal{(2)}$(A^{m+dec(B)},B)$, where $dec(B)$ is the decimal value of $B$.
\end{Def}

An $n$-dimensional hierarchical hypercube $H\!H\!C_{n}$ is a $(m +1)$-regular bipartite graph of $2^{n}$ nodes, where $n =2^{m}+m.$ Clearly, $H\!H\!C_{n}$ is triangle-free and consists of $2^{2^{m}}$ clusters, and each cluster is isomorphic to the $m$-dimensional hypercube $Q_{m}.$ The cross edges between the clusters have the property that every node is incident to exactly one of them (i.e. a perfect matching), any two clusters of different classes have exactly one cross edge between them. An example is shown in Fig. 7 (where $m=2$ and $n=6$).

\begin{figure}[!ht]
 \begin{center}
 \resizebox*{5in}{!}{
 \includegraphics{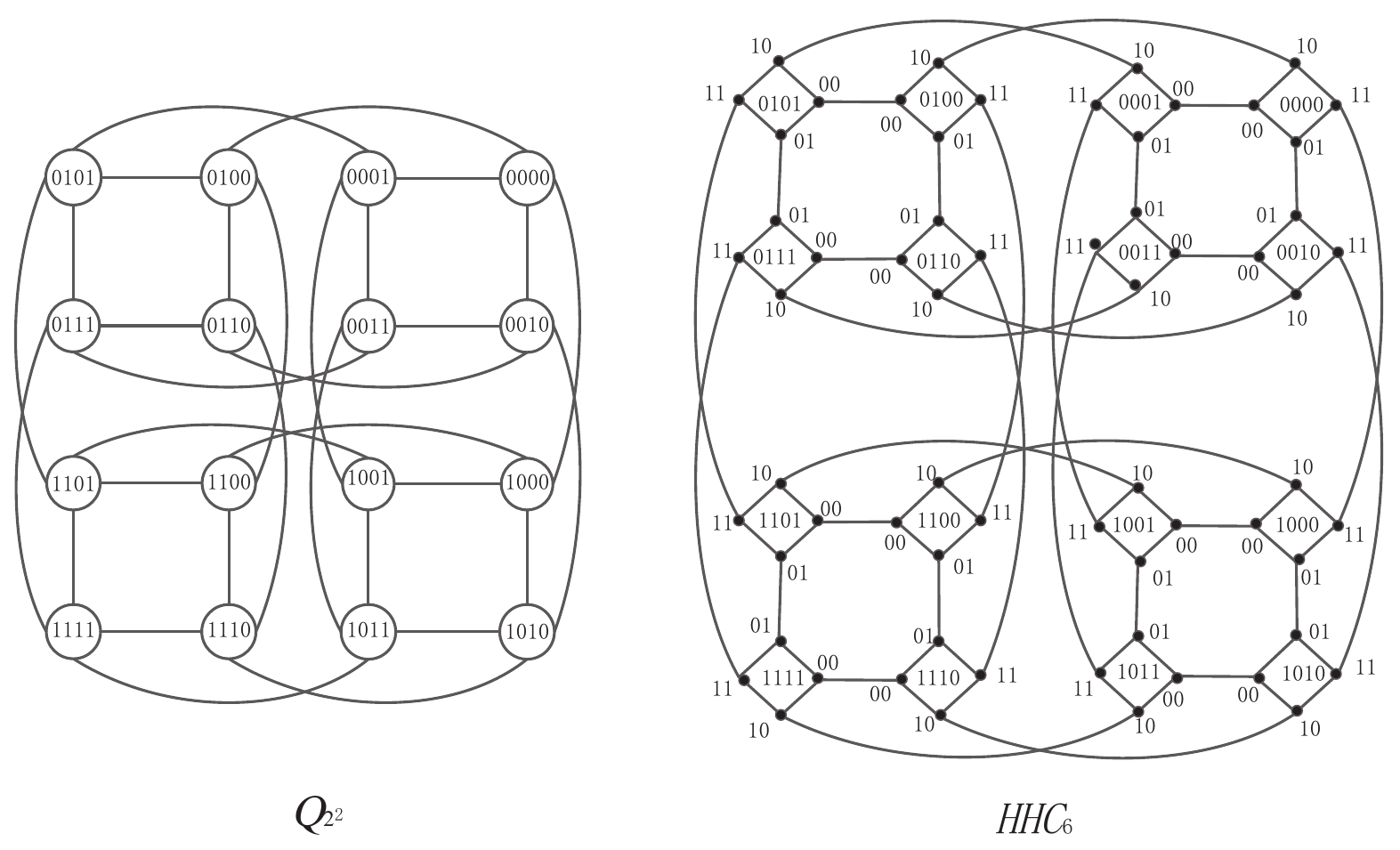}}
 \caption{The illustration of $H\!H\!C_{6}$.}
 \end{center}
\end{figure}

\begin{lem}\textnormal{~\cite{lxyr18}}\label{lem6}
For any integers $m\geq 5$, $n=2^{m}+m$, and $1\leq h\leq m$, let $S\subseteq V(H\!H\!C_{n})$ with $|S|\leq (m+1)h-\frac{(h-1)(h+2)}{2}-1$. If $H\!H\!C_{n}-S$ is disconnected, then $H\!H\!C_{n}-S$ has a large connected component containing at least $2^{n}-|S|-(h-1)$ nodes.
\end{lem}

\begin{thm}
Let $m\geq 5,$ $n=2^{m}+m$ and $1\leq h\leq m-2$, then the $(h+1)$-component diagnosability of $H\!H\!C_{n}$ is $c t_{h+1}(H\!H\!C_{n})=(h+1)m-\frac{h(h+1)}{2}+1$ under the PMC model and MM$^{*}$ model.
\end{thm}
\begin{proof}
By the definition of $H\!H\!C_{n}$, there exists a perfect matching between clusters in $H\!H\!C_{n}$ and $|V(H\!H\!C_{n})|=2^{n}$. $H\!H\!C_{n}$ consists of $2^{2^{m}}$ clusters, denoted as $C_{1}, C_{2}, \ldots, C_{2^{2^{m}}}$, and each cluster is isomorphic to $Q_{m}$.  Since $C_{1}\cong Q_{m}$, let $C_{1}=G(X_{1}, Y_{1}; M)$ with $X_{1}\cong Q_{m-1}^{0}$ and $Y_{1}\cong Q_{m-1}^{1}$. Choosing an arbitrary node $v$ in $X_{1}$ and a node set $A=\{v_{1}, v_{2}, \ldots, v_{h}, v_{h+1}\} \subseteq N_{X_{1}}(v)$, we can obtain that $deg_{H\!H\!C_{n}}(x)=m+1$ for any node $x\in A \cup\{v\}$. For $m\geq 5$, $|V(H\!H\!C_{n})|=2^{n}> 2^{2m}$. Since $C_{1}\cong Q_{m}$, by Lemma \ref{lem30}, for any two nodes of $C_{1}$, if they have common neighbors at all, then they have exactly two common neighbors in $C_{1}$. We have $|N_{C_{1}}(v_{i_{1}})\cap N_{C_{1}}(v_{i_{2}})|=2$ $(1\leq i_{1}< i_{2}\leq h+1)$ and $|N_{C_{1}}(v_{i_{1}})\cap N_{C_{1}}(v_{i_{2}})\cap \cdots \cap N_{C_{1}}(v_{i_{k}})|=1$ $(1\leq i_{1}< i_{2}<\cdots < i_{k}\leq h+1$ and $k\geq 3)$. By the definition of $H\!H\!C_{n}$, the external neighbours of any pair vertices are in the different clusters. Then $|N_{H\!H\!C_{n}-C_{1}}(v_{i_{1}})\cap N_{H\!H\!C_{n}-C_{1}}(v_{i_{2}})|=0$ $(1\leq i_{1}< i_{2}\leq h+1)$. Thus, $|N_{H\!H\!C_{n}}(v_{i_{1}})\cap N_{H\!H\!C_{n}}(v_{i_{2}})|=2$ $(1\leq i_{1}< i_{2}\leq h+1)$ and $|N_{H\!H\!C_{n}}(v_{i_{1}})\cap N_{H\!H\!C_{n}}(v_{i_{2}})\cap \cdots \cap N_{H\!H\!C_{n}}(v_{i_{k}})|=1$ $(1\leq i_{1}< i_{2}<\cdots < i_{k}\leq h+1$ and $k\geq 3)$. Then $H\!H\!C_{n}$ satisfies the condition (a) of Theorem \ref{the1}. Moreover, by Lemma \ref{lem6}, the structure of $H\!H\!C_{n}$ satisfies the condition (b) of Theorem \ref{the1}. Thus, by Theorem \ref{the1}, $c t_{h+1}(H\!H\!C_{n})=(h+1)(r-1)-\frac{h(h+1)}{2}+1=(h+1)m-\frac{h(h+1)}{2}+1$ under the PMC model and MM$^{*}$ model.
\end{proof}

\subsection{Cayley graphs generated by transposition trees \textnormal{(}except star graphs\textnormal{)}}
There are some studies on Cayley graphs generated by transposition trees~\cite{elns07, whj10}. The definition and some available properties of Cayley graphs generated by transposition trees are presented as follows.

\begin{Def}\textnormal{~\cite{whj10}}\label{def7}
Let $\Gamma$ be a finite group, and let $\Delta$ be a subset of $\Gamma$ such that $\Delta$ has no identity element. The directed Cayley graph $\Gamma(\Delta)$ is defined as follows: node set $\Gamma$ and arc set $\{(u, uv)|u\in \Gamma, v\in \Delta\}$. If for each $v\in \Delta$, we also have its inverse $u^{-1}\in \Delta$, then we say that this Cayley graph is an undirected Cayley graph, and every undirected Cayley graph is node-transitive. Denote by $Sym(n)$ the group of all permutations on $\{1, 2, \ldots, n\}$. Let $(l_{1}l_{2}\cdots l_{n})$ be a permutation, which is called a transposition, denotes the permutation that swaps the objects at position $i$ and $j$, that is $(l_{1}\cdots l_{i}\cdots l_{j}\cdots l_{n})(ij)=(l_{1}\cdots l_{j}\cdots l_{i}\cdots l_{n})$. Let $H$ be a set of transpositions, we call $G(H)$ the transposition generating graph, where the node set of $G(H)$ is the set of all $n!$ permutations on $\{1, 2, \ldots, n\}$ and edge set $E(G(H))=\{(ij)|(i,j)\in H\}$.
\end{Def}

\begin{figure}[!ht]
 \begin{center}
 \resizebox*{5in}{!}{
 \includegraphics{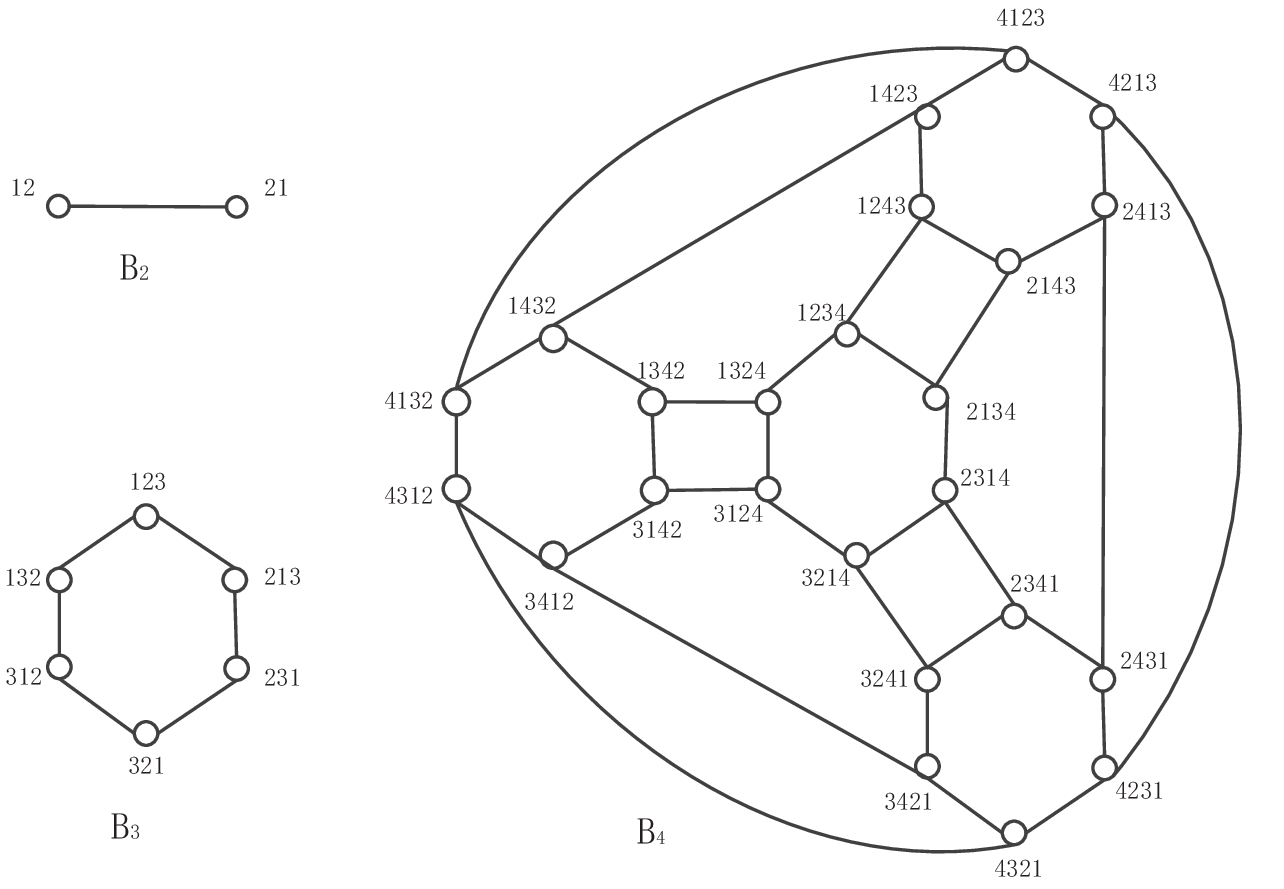}}
 \caption{The illustration of $B_{2}$, $B_{3}$ and $B_{4}$.}
 \end{center}
\end{figure}

When $G(H)$ is a tree, we call the corresponding transposition generating graph a transposition tree. If the
transposition tree is a path, then Cayley graphs are also called bubble-sort graphs. Let $\Gamma_{n}(H)$ be a Cayley graph generated by a transposition tree $G(H)$, where $G(H)$ is not isomorphic to $K_{1,n-1}$.  It is easy to see from the definition that $\Gamma_{n}(H)$ is triangle-free and  is an $(n-1)$-regular graph with $n!$ nodes. Fig. 8 shows bubble-sort graphs $B_{2}$, $B_{3}$ and $B_{4}$.

\begin{lem}\textnormal{~\cite{elns07}}\label{lem7}
For any integers $n\geq 4$ and $1\leq h\leq 3$, let $S\subseteq V(\Gamma_{n}(H))$ with $|S|\leq (n-1)h-\frac{(h-1)(h+2)}{2}-1$. If $\Gamma_{n}(H)-S$ is disconnected, then $\Gamma_{n}(H)-S$ consists of one large component, and remaining small components with at most $h-1$ nodes in total.
\end{lem}

\begin{thm}
Let $n\geq 6$, the $3$-component diagnosability of $\Gamma_{n}(H)$ is $c t_{3}(\Gamma_{n}(H))=3n-8$ under the PMC model and MM$^{*}$ model.
\end{thm}
\begin{proof}

Let $v=l_{1}l_{2}\cdots l_{n}$ and $v_{i}=\big(l_{1}l_{2}\cdots l_{n}\big)\big((n-2i+1)(n-2i+2)\big)$ with $1\leq i\leq 3$. Moreover, let $A=\{v_{1},v_{2},v_{3}\}$. By the definition of $\Gamma_{n}(H)$, we have that $N_{\Gamma_{n}(H)}(v_{i_1})\cap N_{\Gamma_{n}(H)}(v_{i_2})=\{v,v_{i_1}\big((n-2i_2+1)(n-2i_2+2)\big)\}$ with $1\leq i_1<i_2\leq 3$, $N_{\Gamma_{n}(H)}(v_{1})\cap N_{\Gamma_{n}(H)}(v_{2})\cap N_{\Gamma_{n}(H)}(v_{3})=\{v\}$, and $N_{\Gamma_{n}(H)}(v)\cap N_{\Gamma_{n}(H)}(v_i)=\emptyset$. Then $\Gamma_{n}(H)$ satisfies the condition (a) of Theorem \ref{the1}. Moreover, by Lemma \ref{lem7}, the structure of $\Gamma_{n}(H)$ satisfies the condition (b) of Theorem \ref{the1}. Thus, by Theorem \ref{the1}, $ct_{3}(\Gamma_{n}(H))=3n-8$ under the PMC model and MM$^{*}$ model.
\end{proof}

\begin{thm}
Let $n\geq 8$, the $4$-component diagnosability of $\Gamma_{n}(H)$ is $c t_{4}(\Gamma_{n}(H))=4n-13$ under the PMC model and MM$^{*}$ model.
\end{thm}
\begin{proof}
Let $v=l_{1}l_{2}\cdots l_{n}$ and $v_{i}=\big(l_{1}l_{2}\cdots l_{n}\big)\big((n-2i+1)(n-2i+2)\big)$ with $1\leq i\leq 4$. Moreover, let $A=\{v_{1},v_{2},v_{3},v_{4}\}$. By the definition of $\Gamma_{n}(H)$, we have that $N_{\Gamma_{n}(H)}(v_{i_1})\cap N_{\Gamma_{n}(H)}(v_{i_2})=\{v,v_{i_1}\big((n-2i_2+1)(n-2i_2+2)\big)\}$ with $1\leq i_1<i_2\leq 4$, $N_{\Gamma_{n}(H)}(v_{i_{1}})\cap N_{\Gamma_{n}(H)}(v_{i_{2}})\cap \cdots \cap N_{\Gamma_{n}(H)}(v_{i_{k}})=\{v\}$ with $1\leq i_{1}< i_{2}<\cdots < i_{k}\leq 4$ and $k\geq 3$, and $N_{\Gamma_{n}(H)}(v)\cap N_{\Gamma_{n}(H)}(v_i)=\emptyset$. Then $\Gamma_{n}(H)$ satisfies the condition (a) of Theorem \ref{the1}. Moreover, by Lemma \ref{lem7}, the structure of $\Gamma_{n}(H)$ satisfies the condition (b) of Theorem \ref{the1}. Thus, by Theorem \ref{the1}, $ct_{4}(\Gamma_{n}(H))=4n-13$ under the PMC model and MM$^{*}$ model.
\end{proof}

\subsection{DQcube}

DQcube, introduced by Hung ~\cite{hung13}, is a novel compound graph, which uses the hypercube as a unit cluster and connects many such clusters by means of a disc-ring graph at the cost that only one additional link is added to any node in each hypercube. In the following, we first present the concept of compound graph.

\begin{Def}\textnormal{~\cite{guo10}}
Given two regular graphs $G$ and $H,$ the compound graph $G(H)$ is constructed by replacing each node of $G$ by a copy of $H$ and replacing each link of $G$ by a link that connects corresponding two copies of $H$.
\end{Def}

Next, we introduce the disc-ring graph.

\begin{Def}\textnormal{~\cite{hung13}}
The disc-ring graph, represented by $D(m, d)$ \textnormal{(}see Fig.\!\ 9\textnormal{(a))}, consists of outer and inner rings, where each ring contains $m$ nodes and $1 \leq d \leq m .$ The node set of $D(m, d)$ is $\{z_{1} z_{2} \mid z_{1} \in\{0,1\} \text { and } z_{2} \in\{0,1, \cdots, m-1\}\}$, when $z_{1} z_{2}$ is a sequence of two integers and is a label of a node. Node $0 z_{2}$ is in the outer ring while node $1 z_{2}$ is in the inner one. For integer $z$ and positive integer $m,$ we define that $z(\bmod \ m)$ denotes the remainder of the division of $z$ by $m.$ For any node $z_{1} z_{2},$ there is a link connecting it to node $z_{1} x$ for $x \in\{(z_{2}+1)(\bmod \ m),(z_{2}-1)(\bmod \ m)\}$. For any node $0 z_{2}$, $0 \leq z_{2} \leq m-1,$ there is a link connecting it to node $1 y$ for $y \in\{z_{2},(z_{2}+1)(\bmod \ m), (z_{2}+2)(\bmod \ m), \cdots,(z_{2}+ d-1)(\bmod \ m)\}.$ For any node $1 z_{2}$, $0 \leq z_{2} \leq m-1,$ there is a link connecting it to node $0y$ for $y \in\{z_{2},(z_{2}-1)(\bmod \ m),(z_{2}-2)(\bmod \ m), \cdots,(z_{2}-d+1)(\bmod \ m)\}$.
\end{Def}

Sort all nodes of hypercube cluster in the ascending order of the node labels. Then the smallest, the second smallest and the largest nodes will be $\alpha_{1}=00 \cdots 00, \alpha_{2}=00 \cdots 01$ and $\beta=$ $11 \cdots 11,$ respectively. Let $\gamma_{i}$ be the $i$th smallest node excluding any one of $\{\alpha_{1}, \alpha_{2}, \beta\}$ for $1 \leq i \leq 2^{n}-3$. According to the construction methodology of $DQ(m, d, n)$ proposed by Hung~\cite{hung13}, Zhang et al.~\cite{zhang20} give the definition as follows.

\begin{Def}\textnormal{~\cite{zhang20}}\label{def8}
The DQcube is characterized by $DQ(m, d, n)$ \textnormal{(}see Fig.\!\ 9\textnormal{(b))} where $1 \leq d \leq m$ and $d+2=2^{n}$. The node-set $V$ is represented as $\{(z_{1} z_{2}, b_{n-1} b_{n-2} \cdots b_{0})\}$ where $z_{1} z_{2}$ is the label of cluster in $D(m, d)$ and $b_{n-1} b_{n-2} \cdots b_{0}$ is the label of the node in $Q_{n} .$ Two nodes $u=(z_{1} z_{2}, b_{n-1} b_{n-2} \cdots b_{0})$ and $v=(z_{1}^{\prime} z_{2}^{\prime}, b_{n-1}^{\prime} b_{n-2}^{\prime} \cdots b_{0}^{\prime})$ are linked if and only if one of the following conditions is satisfied\textnormal{:}

\textnormal{(1)} $z_{1} z_{2}=z_{1}^{\prime} z_{2}^{\prime}$ and $\sum_{i=0}^{n-1}|b_{i}-b_{i}^{\prime}|=1$\textnormal{;}

\textnormal{(2)} $z_{1}=z_{1}^{\prime}, z_{2}^{\prime}=z_{2}+1, b_{n-1} b_{n-2} \cdots b_{0}=\alpha_{1}$ and
$b_{n-1}^{\prime} b_{n-2}^{\prime} \cdots b_{0}^{\prime}=\alpha_{2}$\textnormal{;}

\textnormal{(3)} $z_{1}-z_{1}^{\prime}=1, z_{2}=z_{2}^{\prime}$ and $b_{n-1} b_{n-2} \cdots b_{0}=$
$b_{n-1}^{\prime} b_{n-2}^{\prime} \cdots b_{0}^{\prime}=\beta$\textnormal{;}

\textnormal{(4)} $z_{1}=0, z_{1}^{\prime}=1, z_{2}^{\prime}=z_{2}+i$ and $b_{n-1} b_{n-2} \cdots b_{0}=$
$b_{n-1}^{\prime} b_{n-2}^{\prime} \cdots b_{0}^{\prime}=\gamma_{i}$.
\end{Def}

$DQ(m, d, n)$ is a $(n +1)$-regular bipartite graph of $m2^{n+1}$ nodes and $(n+1)m2^{n}$ edges. Clearly, $D Q(m, d, n)$ is triangle-free and consists of $2 m$ disjoint clusters. Each node in $DQ(m, d, n)$ is associated with an intercluster edge and has exactly one external-neighbor.

\begin{figure}[!ht]
 \begin{center}
 \resizebox*{5in}{!}{
 \includegraphics{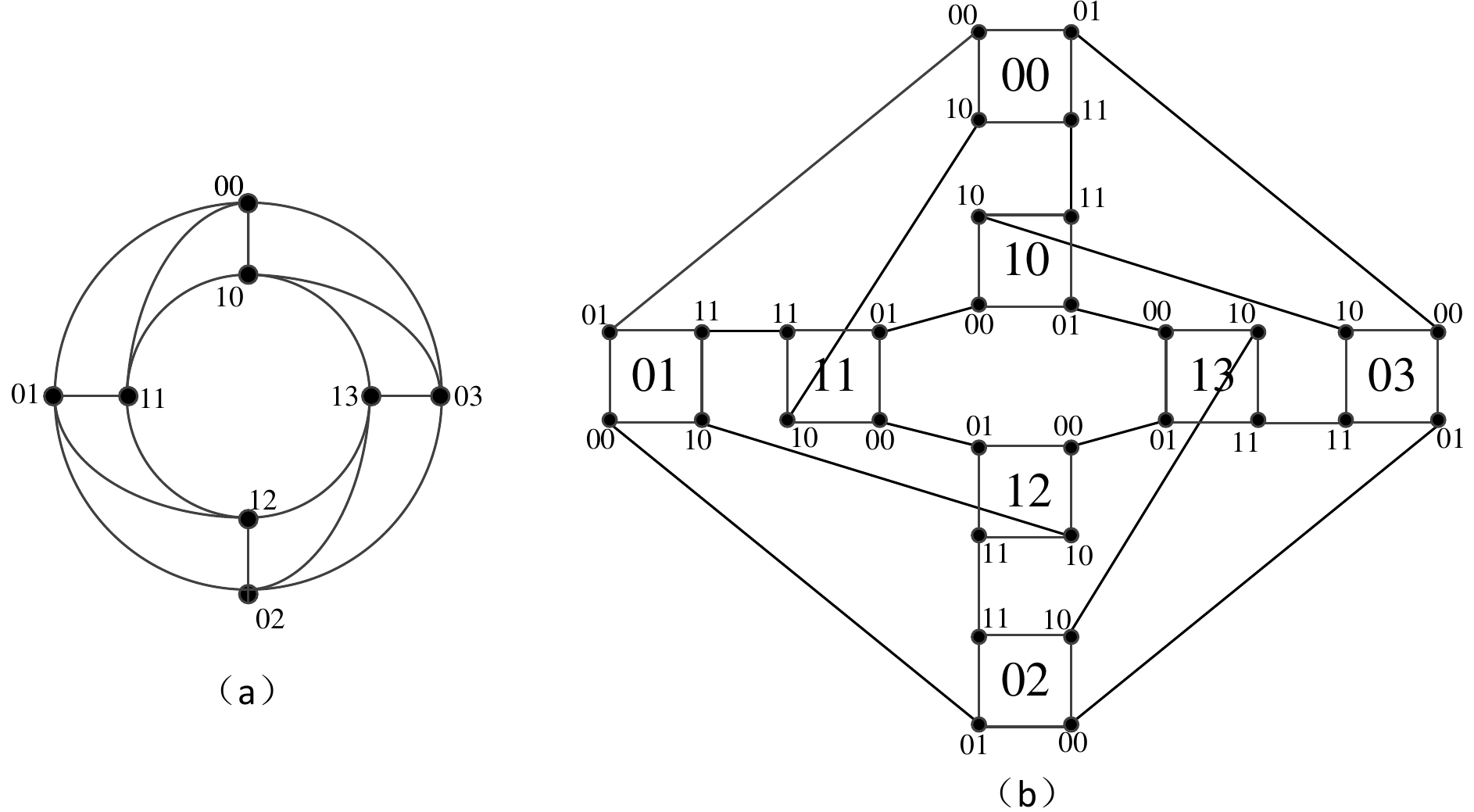}}
 \caption{(a)The disc-ring graphs $D(4, 2)$; (b)The topology of $DQ(4, 2, 2)$.}
 \end{center}
\end{figure}

\begin{lem}\textnormal{~\cite{zhang20}}\label{lem8}
For any integers $1\leq h\leq n+1$, let $S\subseteq V(DQ(m,d,n))$ with $|S|\leq (n+1)h-\frac{(h-1)(h+2)}{2}-1$, then $DQ(m,d,n)-S$ is either connected or it has a component containing at least $m2^{n+1}-|S|-(h-1)$ nodes.
\end{lem}

\begin{thm}\label{x}
Let $n\geq 3$ and $1\leq h\leq n-2$, then the $(h+1)$-component diagnosability of $DQ(m,d,n)$ is $c t_{h+1}(DQ(m,d,n))=(h+1)n-\frac{h(h+1)}{2}+1$ under the PMC model and MM$^{*}$ model.
\end{thm}
\begin{proof}
By the definition of $DQ(m,d,n)$,  there exists a perfect matching between clusters in $DQ(m,d,n)$ and $|V(DQ(m,d,n))|=m2^{n+1}$. $DQ(m,d,n)$ consists of $2m$ clusters, denoted as $C_{1}, C_{2}, \ldots, C_{2m}$, and each cluster is isomorphic to $Q_{n}$.  Since $C_{1}\cong Q_{n}$, let $C_{1}=G(X_{1}, Y_{1}; M)$ with $X_{1}\cong Q_{n-1}^{0}$ and $Y_{1}\cong Q_{n-1}^{1}$. Choosing an arbitrary node $v$ in $X_{1}$ and a node set $A=\{v_{1}, v_{2}, \ldots, v_{h}, v_{h+1}\}  \subseteq N_{X_{1}}(v)$, we can obtain that $deg_{DQ(m,d,n)}(x)=n+1$ for any node $x\in A \cup\{v\}$. For $n\geq 3$, $|V(DQ(m,d,n))|=m2^{n+1}> 2^{2n}$. Since $C_{1}\cong Q_{n}$, by Lemma \ref{lem30}, for any two nodes of $C_{1}$, if they have common neighbors at all, then they have exactly two common neighbors in $C_{1}$. We have $|N_{C_{1}}(v_{i_{1}})\cap N_{C_{1}}(v_{i_{2}})|=2$ $(1\leq i_{1}< i_{2}\leq h+1)$ and $|N_{C_{1}}(v_{i_{1}})\cap N_{C_{1}}(v_{i_{2}})\cap \cdots \cap N_{C_{1}}(v_{i_{k}})|=1$ $(1\leq i_{1}< i_{2}<\cdots < i_{k}\leq h+1$ and $k\geq 3)$. By the definition of $DQ(m,d,n)$, the external neighbours of any pair nodes are in the different clusters. Then $|N_{DQ(m,d,n)-C_{1}}(v_{i_{1}})\cap N_{DQ(m,d,n)-C_{1}}(v_{i_{2}})|=0$ $(1\leq i_{1}< i_{2}\leq h+1)$. Thus, $|N_{DQ(m,d,n)}(v_{i_{1}})\cap N_{DQ(m,d,n)}(v_{i_{2}})|=2$ $(1\leq i_{1}< i_{2}\leq h+1)$ and $|N_{DQ(m,d,n)}(v_{i_{1}})\cap N_{DQ(m,d,n)}(v_{i_{2}})\cap \cdots \cap N_{DQ(m,d,n)}(v_{i_{k}})|=1$ $(1\leq i_{1}< i_{2}<\cdots < i_{k}\leq h+1$ and $k\geq 3)$. Then $DQ(m,d,n)$ satisfies the condition (a) of Theorem \ref{the1}. Moreover, by Lemma \ref{lem8}, $DQ(m,d,n)$ satisfies the condition (b) of Theorem \ref{the1}. Thus, by Theorem \ref{the1}, $c t_{h+1}(DQ(m,d,n))=(h+1)(r-1)-\frac{h(h+1)}{2}+1=(h+1)n-\frac{h(h+1)}{2}+1$ under the PMC model and MM$^{*}$ model.
\end{proof}

\section{Comparison results}

In this section, we will illustrate the advantages of the component diagnosability compared to other famous fault diagnosabilities, including the traditional diagnosability, strong diagnosability, pessimistic diagnosability, conditional diagnosability.

First, we present the definitions of the diagnosis strategies involved in this section. A graph is called \emph{$t$-diagnosable} if the faulty node set has no more than $t$ nodes which can be detected without replacement. Then the \emph{diagnosability} of a graph is the largest $t$ such that it is $t$-diagnosable. There are several ways to extend the concept of diagnosability. For example, the \emph{strong diagnosability} of a graph is the largest number $t$ such that it is strongly $t$-diagnosable, i.e., it is $t$-diagnosable and it will be $(t+1)$-diagnosable when there is no node whose neighbors are all faulty. A system is \emph{$t/t$-diagnosable} if, provided that the number of faulty processors is bounded by $t$, all faulty processors can be isolated within a set of size at most $t$ with at most one fault-free processor mistaken as a faulty one. The \emph{pessimistic diagnosability} of a system $G$ is the maximal number of faulty processors so that the system $G$ is $t/t$-diagnosable. Finally, the \emph{conditional diagnosability} puts constraint on every faulty node set such that they do not contain all neighbors of some nodes. 

Given a general network $G$ defined in Theorem \ref{the1}, we denote the diagnosability by $t(G)$. Then we can obtain that $t(G) \leq \delta(G) \leq r$ \cite{hsu07}, where $r$ is defined in Theorem \ref{the1}. Considering the worst situation, suppose that $\delta(G) = r$. Let $h=1$, by Theorem \ref{the1} we have that the $2$-component diagnosability of $G$ is $2r-2$. Fig. 10 compares the 2-component diagnosability of $G$ with $\delta(G)$. No matter what the value of $r$ is, it can be seen that the $2$-component diagnosability of $G$ is always larger than $\delta(G)$. Since $t(G) \leq \delta(G) \leq r$, we can conclude that the component diagnosability significantly improves the reliability of the interconnection network compared to the traditional diagnosability.

\begin{figure}[!ht]
 \begin{center}
 \resizebox*{5in}{!}{
 \includegraphics{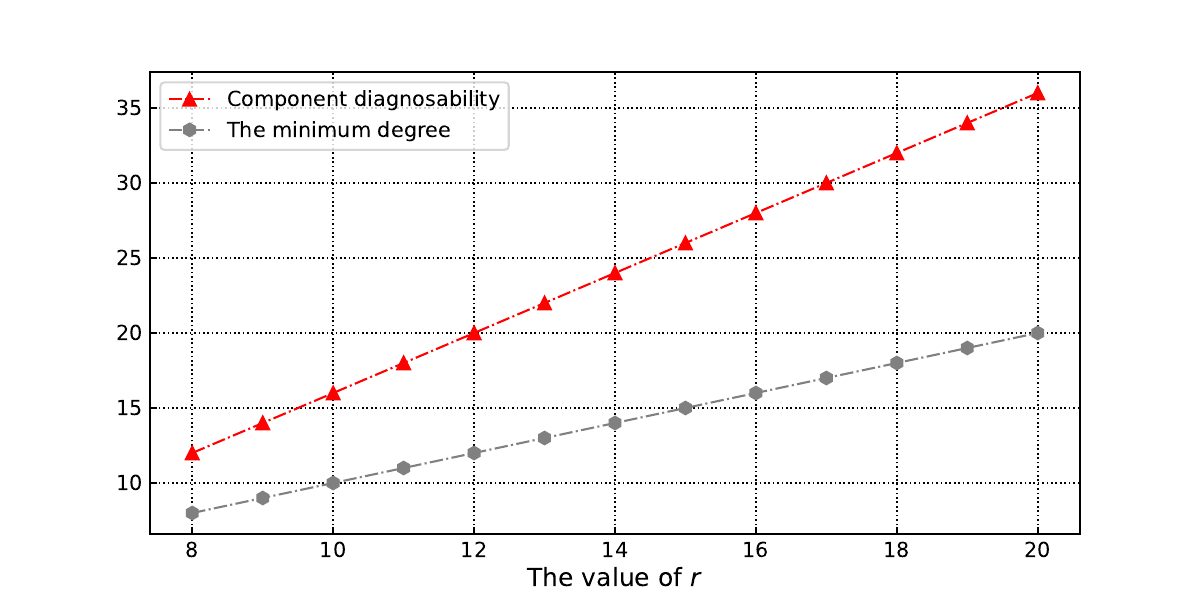}}
 \caption{The illustration of comparisons between the 2-component diagnosability and the minimum degree in a general network $G$.}
 \end{center}
\end{figure}

Several networks mentioned in this paper have similar properties, thus we choose DQcube  as a representative for comparison. Lv et al.~\cite{lv19} have showed that the traditional diagnosability, strong diagnosability, pessimistic diagnosability and conditional diagnosability of DQcube under the PMC model are $n+1$, $n+1$, $2n$, and $4n-3$, respectively. Let $h=4$, by Theorem \ref{x} we can obtain that the $5$-component diagnosability of DQcube is $5n-9$.

Fig. 11 compares the $5$-component diagnosability with strong diagnosability, pessimistic diagnosability, conditional diagnosability in DQcube. When $n$ is small, the $5$-component diagnosability is close to the conditional diagnosability, but far from other fault diagnosabilities. As $n$ increases, the $5$-component diagnosability increases faster than the conditional diagnosability. If let $h\geq 5$, the distance between the $(h+1)$-component diagnosability and the conditional diagnosability will be larger. Therefore, the $(h+1)$-component diagnosability can better evaluate the fault tolerance of the interconnection network.

\begin{figure}[!ht]
 \begin{center}
 \resizebox*{5in}{!}{
 \includegraphics{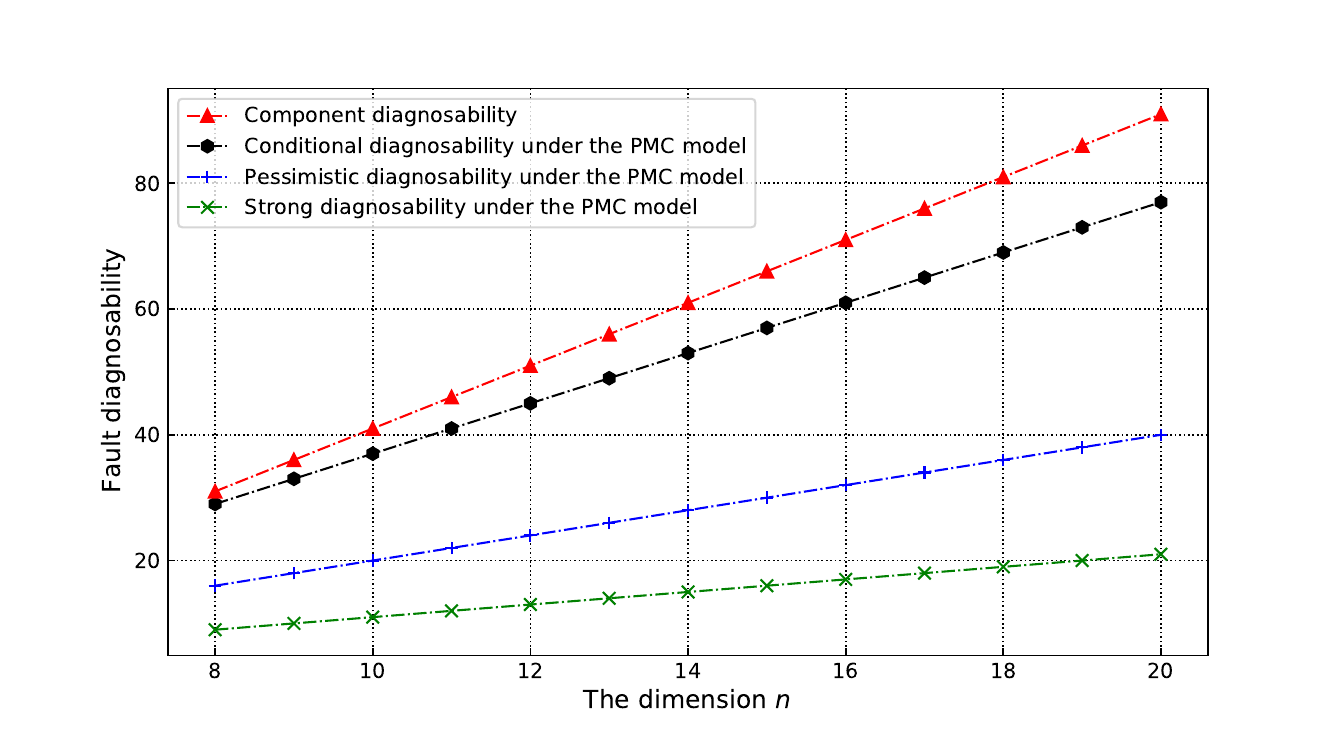}}
 \caption{The illustration of comparisons among the 5-component diagnosability, strong diagnosability, pessimistic diagnosability and conditional diagnosability in DQcube.}
 \end{center}
\end{figure}

\section{Conclusion}
In this paper, we determined the $(h+1)$-component diagnosability $ct_{h+1}(G)$ of general networks under the PMC model and MM$^{*}$ model. As applications, the component diagnosability was explored for some well-known networks, including complete cubic networks, hierarchical cubic networks, generalized exchanged hypercubes, dual-cube-like networks, hierarchical hypercubes, Cayley graphs generated by transposition trees \textnormal{(}except star graphs\textnormal{)}, and DQcube as well. Finally, we made some comparisons to show the advantages of component diagnosability. Future works include evaluating the $(h+1)$-component diagnosability of other interconnection networks and investigating the $(h+1)$-component diagnosability of general networks for general integers $h\geq r-2$.

\section*{Declaration of competing interest}
The authors declare that they have no known competing financial interests or personal relationships that could have appeared to influence the work reported in this paper.

\section*{Acknowledgements}
This work was supported by the National Natural Science Foundation of China (Nos. 62002062, 61872257, 62072109 and U1705262), Education and Scientific Research Project of Young and Middle-aged Teachers of Fujian Provincial Education Department (No. JAT190031). Natural Science Foundation of Fujian Province (Nos. 2018J07005 and 2021J06013).

\end{document}